\IfFileExists{ieeeconf.cls}{%
  \documentclass[letterpaper,10pt,conference]{ieeeconf}
}{%
  \documentclass[letterpaper,10pt,conference]{IEEEtran}
  \typeout{IEEEtran fallback: ieeeconf.cls is not included in the supplied project.}
}
\usepackage{amsmath,amsfonts,amssymb}
\usepackage{xcolor}
\usepackage{graphicx}
\usepackage{cite}
\usepackage{multirow}
\usepackage{booktabs}
\usepackage{tabularx}
\usepackage{float}
\usepackage{xurl}
\usepackage{longtable}
\usepackage{array}
\usepackage[hidelinks]{hyperref}
\newcommand{\CVaR}{\operatorname{CVaR}}

\newcommand{\numNarea}{3}
\newcommand{\numBusA}{118}
\newcommand{\numBusB}{218}
\newcommand{\numBusC}{318}
\newcommand{\numPdCand}{333}
\newcommand{\numDsys}{8550}
\newcommand{\numDreq}{1282.5}
\newcommand{\numdi}{427.5}
\newcommand{\numReqPct}{15}
\newcommand{\numNs}{2000}
\newcommand{\numBandPct}{5}
\newcommand{\numSigmaPct}{2}
\newcommand{\numRiskPct}{5}

\newcommand{\numCf}{772.1}
\newcommand{\numCr}{1129.7}
\newcommand{\numCc}{357.6}
\newcommand{\numGainPct}{46.3}

\newcommand{\numCrOne}{1114.3}

\newcommand{\numCrFive}{1129.7}

\newcommand{\numNprod}{6}
\newcommand{\numEps}{1}

\newcommand{\numKb}{2}

\allowdisplaybreaks

\newtheorem{proposition}{Proposition}
\newtheorem{definition}{Definition}
\newtheorem{lemma}{Lemma}
\newtheorem{theorem}{Theorem}

\providecommand{\IEEEoverridecommandlockouts}{}
\providecommand{\overrideIEEEmargins}{}
\IEEEoverridecommandlockouts
\title{\LARGE \bf
Risk-Aware Allocation of Transmission Capacity for Large Loads
}
\author{Shaoze Li, Bohang Fang and Cong Chen% 
\thanks{The authors are with the Thayer School of Engineering, Dartmouth College, Hanover, NH, USA.
{\texttt{\{Shaoze.Li.TH, Bohang.Fang, Cong.Chen\}@dartmouth.edu}}. \textit{Corresponding author:~Cong Chen.}}

\thanks{This research was supported by the Amazon Research Award (Spring 2025). Any opinions, findings, and conclusions or recommendations expressed in this material are those of the authors and do not necessarily reflect the views of the funding agencies.}}
\begin{document}
\maketitle
\thispagestyle{empty}
\pagestyle{empty}

\begin{abstract}

Rapid growth in data centers and other large loads is straining transmission grid interconnection processes. This paper develops a framework to quantify firm transmission grid capacity and additional risk-aware flexible capacity that can be unlocked when large loads accept a predefined level of interruption risk. We design a normalized unmet-request objective to serve large load requests. We prove that, in radial networks, every minimizer of this objective also maximizes aggregate interconnected capacity, and we show numerically that this property nearly holds on meshed networks. To allocate the firm and flexible transmission capacities among competing large loads, we use a simultaneous ascending auction (SAA) over products differentiated by {\em capacity, risk level, and bus location}. When large loads have additive, symmetric concave, unit-demand, and $K_b$-demand valuations of the firm and flexible capacity products, we show that the gross substitutes property holds, which supports SAA convergence to a competitive equilibrium. A numerical study on the IEEE 73-bus system shows that, at a 1\% risk setting, flexible capacity increases the total network capacity by 44.3\% relative to the firm capacity baseline, while the SAA reaches a competitive equilibrium.

\end{abstract}

\section{Introduction}
The rapid growth of large electricity loads is creating increasing pressure on transmission grid interconnection \cite{wecc_large_load_2025,ferc2025colocation}. AI-driven data centers are an important and rapidly growing example of such loads. In the United States, data centers accounted for approximately 4.4\% of total electricity consumption in 2023, with this share expected to reach 6.7\%--12\% by 2028 \cite{shehabi2024datacenter}. Because grid expansion is slow and capital-intensive, a structural mismatch is emerging between rapidly growing large-load interconnection requests and limited transmission capacity, leading to prolonged interconnection queue times for large-load projects \cite{ferc2025colocation}.

One way to accelerate large-load interconnection is to distinguish between traditional {\em firm capacity} and {\em flexible capacity}. Firm capacity refers to withdrawal capacity that can be supported without curtailment under the prescribed operating conditions, whereas flexible capacity is additional withdrawal capacity that can be curtailed under stressed conditions subject to predefined reliability or risk limits. Historically, grid operators have allocated firm capacity by reserving transmission network capacities for worst-case peak demand, resulting in substantial under-utilization of the grid \cite{gu2026role}. Some large loads can provide the operational flexibility to unlock latent grid capacity. Data centers are an important example: Google has committed to providing 1~GW of demand response capacity \cite{google2026demandresponse}, enabled by workloads such as batch processing and AI model training that can tolerate low-probability interruptions without compromising overall service quality \cite{chen2020internet}. By leveraging low-probability load interruptions, grid operators can allocate substantial flexible transmission capacity \cite{norris2025}, thereby accelerating large-load interconnection and improving grid utilization.

Extensive literature has explored how the existing grid can accommodate growing large-load demand. Recent studies have shown that flexible connections can substantially increase load hosting capacity and accelerate interconnection in distribution networks \cite{gu2026role,gu2026flexible}, while flexibility-aware siting can expand feasible data-center interconnection opportunities in transmission systems \cite{kim2026flexibility}. Collectively, these studies highlight the potential of flexible and risk-aware operation to unlock untapped grid capacity. However, limited attention has been given to the joint challenge of quantifying and allocating firm and flexible capacity in transmission networks under uncertainty.

Allocating scarce network capacity among competing large-load developers poses a further challenge. Existing queue-based regimes process interconnection requests chronologically but can lead to speculative requests, withdrawals, and repeated restudies \cite{wecc_large_load_2025}. The consequences are visible in ERCOT, which tracked over 474 GW of
large-load requests as of June 2026---roughly five times its record peak
demand, with about 90\% associated with data centers
\cite{abbott2026audit}---while only about 4 GW was observed consuming
power \cite{ercot2026overview}. To unlock transmission network capacity, mitigate interconnection queue delays, and address the competition among large-load developers, we therefore introduce a market-based mechanism that allocates firm and flexible transmission capacities among competing large-load developers with heterogeneous capacity valuations. Specifically, we adopt the Simultaneous Ascending Auction (SAA), a well-established multi-item auction in public-sector resource allocation, most notably in spectrum auctions \cite{milgrom2000saa,cramton1997fcc,cramton2002spectrum}. SAA can discover prices and allocate scarce withdrawal capacity according to bidders' heterogeneous valuations, thereby improving resource allocation efficiency \cite{milgrom2000saa}.

Our primary contributions are two-fold.

We propose a normalized unmet-request objective to quantify firm and flexible
transmission capacity for large-load interconnection. For radial networks, we prove that every minimizer of this objective also maximizes aggregate interconnected capacity; on meshed networks, the computed flexible allocations attain 97.2\% of their respective maximum aggregate capacity.

Second, we design an efficient auction-based market mechanism to allocate capacity products. Each bidding product is characterized by three attributes: {\em (capacity, risk level, location)} for transmission network access. We apply the SAA to allocate these firm and flexible capacity products to competing large-load developers. Furthermore, we identify four practically interpretable valuation structures for these capacity products---additive, symmetric concave, unit-demand, and $K_b$-demand valuations---and prove that each satisfies the gross substitutes condition. Thus, SAA with straightforward bidding converges to a competitive equilibrium of the capacity allocation.

The paper is organized as follows. We begin by detailing the firm- and flexible-capacity quantification for transmission networks in Section \ref{section2}. Next, we introduce the SAA and its application to our problem in Section \ref{section3}. Finally, we provide a numerical case study in Section \ref{section4}. A summary of the key symbols used in this paper is provided in Table \ref{notations}. Boldface lowercase letters denote column vectors, e.g., $\boldsymbol{x} = (x_1, \dots, x_n)^\top$. For any scalar $z$, we define $z^+ := \max\{z, 0\}$ and $z^- := z^+ - z$. These definitions extend to vectors and matrices, where the operations are applied element-wise. For any set $N$, let $|N|$ denote its cardinality.
\section{Transmission Capacity Quantification}\label{section2}
In this section, we first determine the firm capacity $c^f_i$ for large-load interconnection requests at bus $i$, and then unlock the additional flexible capacity $c^c_i$ using the risk-aware model. We focus on how much large-load capacity can be accommodated by the existing transmission capacity. Generation cost, generation capacity limits, and grid-expansion costs are not considered in the present framework; the associated economic and planning analysis is left for future research.
This simplification allows us to focus on transmission limitations and unlock latent grid capacity, rather than mixing generation capacity limits and economic models together.

\begin{table}[htbp]
    \centering
    \small
    \vspace{-1.5em}
    \caption{Summary of Notations}
    \label{notations}
    
    \begin{tabularx}{\columnwidth}{c X}
        \hline
        \textbf{Symbol} & \textbf{Description} \\
        \hline
        $N$ & Transmission network bus index set. \\
        $J$ & Transmission network line index set. \\
        $c^f_i\in \mathbb{R}_+$
        & The {\em firm capacity} for large loads at bus $i$. \\
        $c^r_i\in \mathbb{R}_+$
        & The {\em flexible capacity} for large loads with an outage risk level $r$ at bus~$i$. \\
        $c^c_i\in \mathbb{R}_+$
        & The incremental flexible capacity by accepting the outage risk $r$
        (defined as $c^r_i-c^f_i$). \\
        ${l}_i\in \mathbb{R}$
        & The existing random electricity net demand at bus $i$, bounded by
        $[\underline{l}_i,\overline{l}_i]$. \\
        $\overline{q}_i\in \mathbb{R}_+$
        & Power withdrawal limitation at bus $i$. \\
        $d_i\in \mathbb{R}_+$
        & Total new load interconnection request quantity at bus $i$. \\
        $\boldsymbol{S}\in \mathbb{R}^{|J|\times|N|}$
        & Shift-factor matrix of the DC power flow model. \\
        $\boldsymbol{\overline{b}},\boldsymbol{\underline{b}}\in \mathbb{R}^{|J|}$
        & Transmission network line thermal limits. \\
        \hline
    \end{tabularx}
    \vspace{-1.5em}
\end{table}

\subsection{{Robust Optimization for Firm Capacity}}\label{section_ro}

In this subsection, we determine the firm withdrawal capacity available for large-load interconnection while ensuring system security under physical network constraints.
\begin{subequations}\label{ro}
\begin{align}
\min_{\boldsymbol{c^f}\ge \boldsymbol{0}} \quad
& \sum_{i \in N\setminus \Omega}
\left(
\frac{d_i-c^f_i}{d_i}
\right)^2
\label{ro1a}\\[4pt]
\text{s.t.}\quad
& \boldsymbol{c^f}+\boldsymbol{\overline{l}}
\le \boldsymbol{\overline{q}},
\label{ro1b}\\
& \boldsymbol{\underline{b}}
\le
\boldsymbol{S(c^f+l)}
\le
\boldsymbol{\overline{b}},
\quad
\forall~
\boldsymbol{l}\in
[\boldsymbol{\underline{l}},\boldsymbol{\overline{l}}].
\label{ro1c}
\end{align}
\end{subequations}

Here, the objective is to minimize the normalized unserved demand ratio at buses with load interconnection requests to maximize overall large-load connectivity. This objective accounts for heterogeneous demand
levels, preventing locations with smaller demands from being
overserved.  Let $\boldsymbol{d}$ denote the aggregated withdrawal capacity demand vector for all large loads seeking connectivity, where $d_i-c_i^f$ represents the unserved demand at bus $i$. We define $\Omega$ as the set of buses without any interconnection requests, i.e., $d_i=c_i^f=0$ for all $i\in\Omega$. We assume that $d_i$ is finite and $d_i>c_i^f$ for the scarcity of network firm capacity, based on the fact that there is a long queue of large-load projects waiting for network interconnection. Define $\boldsymbol{l}$ as a random vector representing the existing net electricity demand, accounting for uncertainties in both renewable generation and base load. Constraint~\eqref{ro1b} enforces nodal withdrawal limits for the worst case, e.g., the maximum capacity of a local transformer. For buses without specific restrictions, the withdrawal limit $\overline{q}_i$ is set to infinity. Note that we omit power injection limits from this formulation; if the maximum existing injection already satisfies grid limits, such constraints remain satisfied when additional withdrawal capacity is introduced.

In this work, we employ a DC power flow model on the transmission networks adapted from \cite{low2022power}. Constraint~\eqref{ro1c} enforces physical network upper and lower limits for all possible total withdrawal capacities. Using the box structure of the uncertainty set and the reformulation in \cite{chen2024wholesale}, constraint~\eqref{ro1c} {is equivalent to the following linear constraints:}
\begin{subequations}\label{ro_reform}
\begin{align}
\boldsymbol{S}\boldsymbol{c^f}
&\le
\boldsymbol{\overline{b}}
-
\left(
\boldsymbol{S}^+\boldsymbol{\overline{l}}
-
\boldsymbol{S}^-\boldsymbol{\underline{l}}
\right),
\label{eq:robust_reform_ub}\\
\boldsymbol{S}\boldsymbol{c^f}
&\ge
\boldsymbol{\underline{b}}
-
\left(
\boldsymbol{S}^+\boldsymbol{\underline{l}}
-
\boldsymbol{S}^-\boldsymbol{\overline{l}}
\right).
\label{eq:robust_reform_lb}
\end{align}
\end{subequations}
Such a representation transforms \eqref{ro} into a convex quadratic program with linear constraints, enabling the computation of firm capacity $\boldsymbol{c^f}$ while ensuring network security.

The demand-normalized objective in \eqref{ro} is designed to account for heterogeneous large-load interconnection requests, rather than allocating capacity according to the total number of MW that can be accommodated. This naturally raises an important question: does balancing the fulfillment of heterogeneous requests reduce the aggregate hosting capacity of the network? The following structural result shows that, for radial networks, this tradeoff does not arise. The solution of \eqref{ro} simultaneously minimizes the normalized unmet requests and attains the maximum aggregate firm capacity permitted by the same network constraints.  For meshed networks, our
simulation results show that the difference is small.

\begin{proposition}\label{merge}
If the network is radial, the optimal solution to problem~\eqref{ro} is also optimal for the following total capacity maximization problem:
\begin{equation}\label{eq:one}
\max_{\boldsymbol{c^f}\ge\boldsymbol{0}}
\quad
\sum_{i\in N\setminus\Omega}c_i^f
\quad
\text{s.t.}
\quad
\eqref{ro1b}~\&~\eqref{ro1c}.
\end{equation}
\end{proposition}
The proof is provided in Appendix~\ref{proofth1}. 

\subsection{Risk-aware Model for Flexible Capacity}

Problem~\eqref{ro} derives a withdrawal capacity that guarantees network security across all possible scenarios. However, by strictly accounting for low-probability extreme events, the robust approach inevitably yields overly conservative capacity estimates. Consequently, this can restrict the grid's available hosting capacity for rapidly growing large-load interconnection requests. From the perspective of large-load operators, some projects may be willing to tolerate a marginal risk of power interruption in exchange for expedited grid interconnection, provided that their operations can accommodate such curtailment. Data centers are one important example: non-real-time or flexible computing workloads can be deferred during periods of grid stress and  low-probability power interruptions
can be mitigated by co-located or rented backup generators.

Consequently, we introduce the following risk-aware model to extract extra flexible capacity. The conditional value at risk (CVaR) is adopted as our risk measure because it not only controls the probability and the magnitude of constraint violations, but also requires no restrictive assumptions on the underlying uncertainty distribution.

The firm capacity $\boldsymbol{c^f}$ determined above is treated as a fixed baseline, and the flexible-capacity calculation is required to retain at least this amount of capacity.
\begin{subequations}\label{cvar}
\begin{align}
\min_{\boldsymbol{c^r}\ge\boldsymbol{c^f}}
\quad
&\sum_{i\in N\setminus\Omega}
\left(
\frac{d_i-c_i^r}{d_i}
\right)^2
\\[4pt]
\text{s.t.}\quad
&
\textbf{CVaR}_{\boldsymbol{\alpha}}
[\boldsymbol{c^r}+\boldsymbol{l}]
\le
\boldsymbol{\overline{q}},
\label{cvar_1}\\
&
\textbf{CVaR}_{\boldsymbol{\alpha}}
[\boldsymbol{S(c^r+l)}]
\le
\boldsymbol{\overline{b}},
\label{cvar_2}\\
&
\textbf{CVaR}_{\boldsymbol{\alpha}}
[-\boldsymbol{S(c^r+l)}]
\le
-\boldsymbol{\underline{b}},
\label{cvar_3}
\end{align}
\end{subequations}
where $\boldsymbol{r}=\boldsymbol{1-\alpha}$ represents the outage risk level for the flexible capacity. Specifically, the firm capacity \(c_i^f\) acts as a strictly must-serve load,
whereas the incremental flexible capacity \(c_i^r-c_i^f\) acts as an
interruptible load. This risk level indicates the probability of network constraint violations in \eqref{cvar_1}--\eqref{cvar_3}, during which the flexible capacity may face forced power interruption under extreme realizations of the random net demand $\boldsymbol{l}$. Compared to the robust withdrawal model in \eqref{ro}, model (4) imposes CVaR constraints on withdrawal limitation and network limit violations in \eqref{cvar_1}--\eqref{cvar_3}. Rather than demanding zero violations under all cases, CVaR constraints ensure that the probability of forced power interruptions for the flexible capacity adheres to the predefined interruption risk tolerance $\boldsymbol{r}$. We solve~\eqref{cvar} with a tractable scenario-based reformulation following \cite{chen2024wholesale}.

Similar to the firm-capacity case, we prove that using the demand-normalized objective does not reduce the maximum aggregate capacity under the risk-aware formulation.

\begin{proposition}\label{pro2} For the radial network, the optimal solution for problem~\eqref{cvar} is also optimal for the following total risk-aware capacity maximization problem:
\vspace{-0.5em}
\begin{equation}\label{eq:one_cvar}
\max_{\boldsymbol{c^r}\ge\boldsymbol{c^f}}
\quad
\sum_{i\in N\setminus\Omega}c_i^r
\quad
\text{s.t.}
\quad
\eqref{cvar_1}~-~\eqref{cvar_3}.
\end{equation}
\end{proposition}
The proof is provided in Appendix~\ref{app:flexproof}.

\section{Value-Based Allocation of Capacity Products}\label{section3}
Section~\ref{section2} determines the firm and flexible capacity products that can be supported by the existing transmission network. Once these capacity products are determined, however, a separate challenge remains: how should scarce network-access capacity be allocated among competing large-load projects? Existing queue-based procedures primarily prioritize projects according to application order and therefore do not directly account for the heterogeneous values that different large-load developers place on capacity. This motivates a value-based allocation mechanism. We adopt the SAA because the capacity-allocation problem naturally involves multiple products differentiated by capacity, reliability, and location. SAA allows these products to be priced simultaneously, enables bidders to switch among alternative products as relative prices rise, and provides an iterative price-discovery process without bidders' private valuations in advance.

Our contribution is not the SAA mechanism itself, but the theoretical analysis of when it provides a well-behaved allocation for this transmission capacity. We identify four practically interpretable valuation structures for large-load developers---including additive, symmetric concave, unit-demand, and $K_b$-demand valuations---and prove that they satisfy the gross substitutes property.  Building on this result, we establish that a competitive equilibrium exists for these valuation classes. Thus, the proposed allocation framework provides not only price discovery, but also equilibrium  guarantees.

\subsection{Simultaneous Ascending Auction}\label{SecSAA}

The SAA is a multi-round auction for multiple items. In each round, large-load developers submit bids simultaneously for any items in which they are interested. The capacity bidding item is defined as a tuple characterized by three attributes: capacity, risk level, and location. Building on models~\eqref{ro} and~\eqref{cvar} in Section~\ref{section2}, we determine both the firm capacity and an additional flexible capacity with a risk level for each location. Consequently, for each candidate location, we construct two distinct indivisible bidding products:

\begin{enumerate}
    \item \emph{Firm capacity $c_i^f$} represents capacity with zero risk. It is formulated as the tuple $(c_i^f,0,i)$, where $c_i^f$ denotes the firm capacity, $0$ indicates the risk level, and $i$ represents the bus location index.

    \item \emph{Incremental flexible capacity $c_i^c$} represents additional flexible capacity subject to the predefined outage risk $r$. It is formulated as the tuple $(c_i^c,r,i)$, where $c_i^c:=c_i^r-c_i^f$.
\end{enumerate}

It is worth noting that the firm capacity $c_i^f$ and the incremental flexible capacity $c_i^c$ could be partitioned into smaller indivisible units, similar to spectrum auctions \cite{fcc2021auction110}. However, to ensure a clear exposition of the bidding process, we focus on the formulation in which the capacities at each location are not further subdivided in the remainder.

Once these products are constructed, their physical attributes remain fixed during the auction. The SAA determines which bidder receives each capacity product. During any given round of SAA, large-load developers submit their bids based on their private values. Importantly, any new bid must meet a minimum bid increment, which is computed by adding a predetermined bid increment $\epsilon$ to the item's current standing high bid. Once the round closes, the auctioneer posts the bidding results. For each capacity item, the new standing high bid becomes the maximum of the previous standing high bid and the highest new bid in the current round. The bidder who placed the highest bid is recorded as the new standing high bidder.

\subsection{Capacity Valuations and Theoretical Guarantees}\label{sec:value}

Because large-load developers differ in their preferred locations and tolerance for power interruption, we represent their heterogeneous private values through valuations over the capacity, risk, and location attributes of each capacity product.

Consider a market comprising a set $K$ of indivisible capacity products. As defined in Section~\ref{SecSAA}, each product corresponds to either a firm capacity product represented by the tuple $(c_i^f,0,i)$ or an incremental flexible capacity product represented by $(c_i^c,r,i)$.

Assume that there are $m$ large-load bidders participating in SAA to obtain network capacity. For any bidder $b$, the valuation function $V_b(\cdot)$ is defined over all possible capacity product bundles $U\subseteq K$, where $v^b_j$ denotes bidder $b$'s standalone valuation for product $j$. Let $\boldsymbol{p}\in\mathbb{R}^{|K|}_+$ denote a non-negative price vector, where $p_j$ represents the price of product $j\in K$. Assuming quasi-linear preferences, bidder $b$'s net utility from acquiring a bundle $U$ given the price vector $\boldsymbol{p}$ is expressed as $V_b(U)-\sum_{j\in U}p_j$.

The optimal demand correspondence
$D_b(\boldsymbol{p})$ is defined as
\begin{equation}\label{demand}
D_b(\boldsymbol{p})
=
\arg\max_{U\subseteq K}
\big(
V_b(U)-\sum\nolimits_{j\in U}p_j
\big).
\end{equation}

\begin{definition}[Competitive Equilibrium \cite{milgrom2000saa}]\label{defCE}
A set of non-negative prices $p_k\ge0$ for products $k\in K$ constitutes a competitive equilibrium if there exists a feasible allocation $(U_1^*,U_2^*,\dots,U_m^*)$ that is supported by these prices. This supported allocation must satisfy two conditions:
\begin{subequations}\label{CE}
\begin{align}
\begin{split}
V_b(U_b^*)-\sum_{j\in U_b^*}p_j
&\ge
V_b(U)-\sum_{j\in U}p_j,\\
&\quad
\forall~U\subseteq K,\quad b=1,\dots,m,
\end{split}
\label{ce1}\\
\sum_{j\in K}p_j
&=
\sum_{b=1}^{m}\sum_{j\in U_b^*}p_j.
\label{ce2}
\end{align}
\end{subequations}
\end{definition}

The first condition~\eqref{ce1} ensures that the allocated bundle $U_b^*$ maximizes the net utility for bidder $b$ across all possible bundles of $K$, given the price $\boldsymbol{p}$. The second condition~\eqref{ce2} implies that the price of any unallocated product must be zero.

\begin{definition}[Straightforward Bidding\cite{milgrom2000saa}]
A bidder $b$ bids straightforwardly if it bids on a bundle in its demand correspondence $D_b(\boldsymbol{p})$ defined in~\eqref{demand} at each round.
\end{definition}

\begin{definition}[Gross Substitutes\cite{milgrom2000saa}]
A valuation function $V(\cdot)$ satisfies the gross substitutes condition if, for any two price vectors $\boldsymbol{p},\boldsymbol{p}'\in\mathbb{R}^{|K|}_+$ such that $\boldsymbol{p}'\ge\boldsymbol{p}$ and for any preferred bundle $U\in D(\boldsymbol{p})$, there exists a bundle $U'\in D(\boldsymbol{p}')$ such that
$
\{j\in U\mid p'_j=p_j\}\subseteq U'.
$
\end{definition}

The definition states that if a bidder optimally demands a bundle $U$ at a given price vector $\boldsymbol{p}$, and the prices of some products in $U$ increase while the prices of others remain unchanged, the bidder can continue to demand those products in $U$ whose prices do not increase.

The key question is therefore which practically relevant large-load preferences preserve the gross substitutes property and hence the competitive equilibrium guarantee of SAA. {We consider four valuation structures of the large load bidder:}

\begin{enumerate}
\item \label{additivity}
{An additive valuation: $V_b(U)=\sum_{j\in U}v^b_j$.}

\item \label{number}
{Symmetric concave valuation: $V_b(U)=f_b(|U|)$,
where $f_b:\mathbb{R}_+\rightarrow\mathbb{R}_+$ is monotone and concave.}

\item \label{unit}
{Unit-demand valuation:
$V_b(U)
=
\max_{j\in U}v^b_j$.}

\item \label{kdemand}
{$K_b$-demand valuation:
$V_b(U)
=
\max_{\substack{U'\subseteq U\\|U'|\le K_b}}
\sum_{j\in U'}v^b_j$,
where $K_b\in\{1,\dots,|K|\}$.}
\end{enumerate}

{Valuation~\ref{additivity} above} represents a large-load developer that may obtain capacity at multiple locations and values each capacity product largely independently. {Valuation~\ref{number}} represents bidders that mainly value the number of standardized capacity products obtained, with diminishing marginal value for additional products. {Valuation~\ref{unit}} represents a developer that intends to proceed with at most one project, whereas {Valuation~\ref{kdemand}} represents a developer that proceeds with up to $K_b$
projects and selects the most valuable subset of capacity products.

These four valuation structures capture representative large-load preferences while allowing substitution across capacity products as relative prices change. The following proposition shows that all four satisfy the gross substitutes property required for the equilibrium guarantees of SAA.

\begin{proposition}\label{pro:gs}
Under any of Valuations~\ref{additivity}--\ref{kdemand}, each bidder's valuation function satisfies the gross substitutes property.
\end{proposition}

\begin{proposition}\label{epsilonoptimal}
Let $\boldsymbol{p}^*$ be the final standing price vector and $(U_1^*,U_2^*,\dots,U_m^*)$ be the final feasible allocation. Then, if gross substitutes hold for every bidder and all bidders bid straightforwardly, $(\boldsymbol{p}^*,U_1^*,\dots,U_m^*)$ constitutes a competitive equilibrium for modified valuation functions $\hat{V}_b$, which are adjusted for each bidder $b$ as follows:
\begin{equation}\label{modifiedvaluation}
\hat{V}_b(U)
:=
V_b(U)
-
\epsilon|U\setminus U_b^*|,
\end{equation}
where $\epsilon$ is the minimum bid increment of SAA. Furthermore, the final allocation maximizes the aggregate bidder value to within a single bid-increment margin:
\begin{equation}\label{efficient}
\max_{(U_1,U_2,\dots,U_m)}
\sum_{b=1}^{m}V_b(U_b)
-
\sum_{b=1}^{m}V_b(U_b^*)
\le
\epsilon|K|.
\end{equation}
\end{proposition}

\begin{theorem}\label{th:gs}
If the gross substitutes condition holds for every bidder, a competitive equilibrium exists. Under this condition, if all bidders bid straightforwardly, the final SAA outcome converges to a competitive equilibrium as $\epsilon\rightarrow0$.
\end{theorem}

Together, Propositions~\ref{pro:gs}--\ref{epsilonoptimal} and Theorem~\ref{th:gs} provide the theoretical foundation for the proposed value-based capacity allocation. For the four valuation structures identified above, a competitive equilibrium exists; under straightforward bidding, the SAA outcome converges to such an equilibrium as $\epsilon\rightarrow0$, while for a finite bid increment the loss in aggregate bidder value is bounded by $\epsilon|K|$.

The proof of Proposition~\ref{pro:gs} is provided in
Appendix~\ref{proofprogs}. Proposition~\ref{epsilonoptimal} and
Theorem~\ref{th:gs} follow from the established SAA results in
\cite{milgrom2000saa}; details are provided in
Appendix~\ref{app:auctionproof}.
\subsubsection*{Computational Tractability of Straightforward Bidding}

Under straightforward bidding, a large-load developer solves the demand
problem in~\eqref{demand} in each SAA round.
Although the general demand problem in~\eqref{demand} contains $2^{|K|}$ possible bundles, exhaustive enumeration is not required under the four valuation structures considered above. Additive and unit-demand valuations can be optimized in $O(|K|)$ time. For a $K_b$-demand valuation, the bidder sorts the individual net values $v^b_j-p_j$ and selects up to $K_b$ products with the largest positive values. For a symmetric concave valuation, the products are sorted by price, and the bidder compares the surplus of each possible bundle size. Both cases require $O(|K|\log|K|)$ time. Thus, straightforward bidding remains computationally tractable for the same valuation classes.

\section{Simulations}\label{section4}
{We evaluate the proposed framework on a standard transmission
benchmark. The experiments quantify the firm and flexible capacities available
under different risk levels and evaluate the allocation of the resulting capacity
products.}

\subsection{Experimental Setup}

We use the standard PGLib-OPF IEEE RTS 73-bus benchmark \cite{pglib}, a
meshed transmission system with \numNarea{} areas. The fixed background
net-demand profile $\boldsymbol l$ is constructed as the difference between the
nodal load and a fixed generation profile. Specifically, each in-service
generator's PGLib setpoint is multiplied by a common scaling factor
of $1.3$ and then clipped to the generator's original limits. The resulting
generation profile remains fixed throughout the capacity calculations. Existing nodal loads are sampled independently from truncated normal
distributions centered at their nominal values, with a standard deviation of
\numSigmaPct{}\% of the nominal load and truncation at
$\pm\numBandPct{}\%$, using \numNs{} scenarios for \eqref{cvar}. Thus,
the uncertainty in $\boldsymbol l$ is represented through the nodal-load
scenarios, while the generation profile remains fixed across scenarios.

Three candidate large-load locations are selected by a predetermined rule:
the largest-load bus in each area, giving bus IDs \numBusA{}, \numBusB{}, and
\numBusC{}, which carry the same nominal load of \numPdCand{}~MW.
These numbers are bus IDs rather than the total number of buses. Under the RTS
area-based numbering convention, the 100-, 200-, and 300-series bus IDs
correspond to Areas 1, 2, and 3, respectively. The aggregate new request is \numReqPct{}\% of the nominal system load of
\numDsys{}~MW, that is \numDreq{}~MW, and is divided equally so that
$d_i=\numdi$~MW at each location. The main flexible-service setting is
$r=\numRiskPct{}\%$. All optimization models are implemented in Python and
solved using Gurobi 13.0.3.

\subsection{Firm and Flexible Capacity}

\begin{table}[t]
  \centering
  \caption{Firm and incremental flexible large-load capacities with
  $r=\numRiskPct\%$. All quantities in MW.}
  \label{tab:capacity}
  \setlength{\tabcolsep}{6pt}
  \begin{tabular}{lrrrr}
    \toprule
     & Bus 118 & Bus 218 & Bus 318 & Total \\
    \midrule
     $c_i^f$
      & 191.4 & 321.5 & 259.2 & 772.1 \\
     $c_i^c$
      & 165.4 & 74.2 & 117.9 & 357.6 \\
    \bottomrule
  \end{tabular}
  \vspace{-1.5em}
\end{table}

{Let $C^f:=\sum_i c_i^f$, $C^r:=\sum_i c_i^r$, and
$C^c:=\sum_i c_i^c$} denote the aggregate firm, risk-aware, and incremental
flexible capacities over the three candidate buses, respectively.
Table~\ref{tab:capacity} reports the outcome at the three locations. Under
the all-scenario security requirement, the network accommodates
$C^f=\numCf{}$~MW of firm capacity. Accepting a \numRiskPct{}\% interruption
risk on the incremental service raises total admitted capacity to
$C^r=\numCr{}$~MW, so the incremental flexible capacity is
$C^c=\numCc{}$~MW, an increase of \numGainPct{}\% over the firm baseline.
The increments differ across locations even though the three requests are
identical. On this case, the value of a flexible connection is therefore
location-dependent.

Because the benchmark is meshed, Proposition~\ref{merge} does not directly apply. We therefore also solve the
corresponding capacity models by directly maximizing aggregate admitted
capacity. For the firm model, the demand-normalized solution admits
$772.1$~MW, compared with $878.8$~MW under direct aggregate-capacity
maximization, attaining $88.0\%$ of the aggregate maximum. At $r=5\%$, the
corresponding capacities are $1129.7$~MW and $1161.6$~MW, respectively, so
the demand-normalized solution attains $97.2\%$ of the aggregate maximum.
Although exact equivalence does not hold for this meshed benchmark, the resulting
solutions remain close to the aggregate-capacity maxima.

\begin{figure}[t]
  \centering
  \includegraphics[width=0.6\columnwidth]{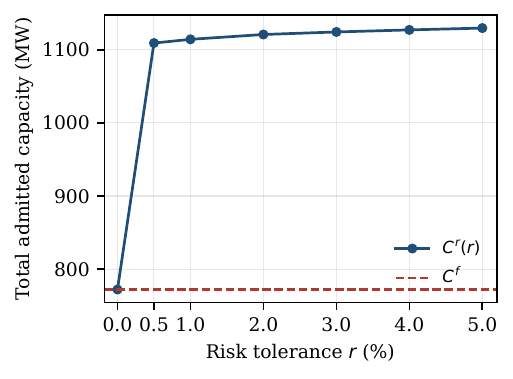}
  \vspace{-1.5em}
  \caption{Total admitted capacity $C^r$ against the prescribed risk level $r$.}
  \label{fig:risk}
  \vspace{-1.5em}
\end{figure}

Fig.~\ref{fig:risk} traces total admitted capacity against the prescribed
risk level. The computed capacity increases over the tested range, from
\numCrOne{}~MW at $r=1\%$ to \numCrFive{}~MW at $r=5\%$. Most of the
additional hosting capacity is obtained by moving from the all-scenario firm
requirement to a small nonzero interruption-risk tolerance, whereas further
relaxation yields diminishing additional capacity.

\subsection{Value-Based Allocation of Capacity Products}

The solved capacities define the supply side of the auction directly: each
location contributes a firm product $(c_i^f,0,i)$ and an incremental flexible
product $(c_i^c,r,i)$, giving $|K|=\numNprod{}$ products whose capacities are
exactly the $c_i^f$ and $c_i^c$ entries of Table~\ref{tab:capacity}. This same
set of products is offered to four large-load bidders under each of the four valuation structures of Section~\ref{section3}. Standalone values
are built from each product's capacity, location, and risk attribute, so
bidders differ in their preferred location and in their sensitivity to
interruption risk. The bid increment is $\epsilon=\numEps{}$, and in each
round, each bidder's straightforward bid is obtained using the
valuation-specific demand procedure in Section~\ref{section3}.

\begin{table}[h]
  \centering
  \vspace{-1.0em}
  \caption{SAA outcomes for the four valuation structures on the
  network-derived capacity products.}
  \label{tab:auction}
  \setlength{\tabcolsep}{3.5pt}
  \begin{tabular}{lrrrc}
    \toprule
    Valuation & Rounds & $V^{\rm SAA}$ & $V^{\rm OPT}$ & CE\\
    \midrule
    Additive & 81 & 346.13 & 346.13 & Yes \\
    Symmetric concave & 106 & 455.57 & 455.57 & Yes \\
    Unit-demand & 60 & 284.32 & 284.32 & Yes \\
    $K_b$-demand & 120 & 346.13 & 346.13 & Yes \\
    \bottomrule
  \end{tabular}
\end{table}
 \vspace{-1em}
Table~\ref{tab:auction} summarizes the outcomes under the four valuation
structures. ``Rounds'' denotes the number of bidding rounds;
$V^{\rm SAA}$ is the aggregate bidder value achieved by the SAA allocation,
and $V^{\rm OPT}$ is the centralized benchmark obtained by maximizing the
aggregate bidder value over all feasible allocations. ``CE'' indicates whether
the finite-$\epsilon$ competitive-equilibrium conditions under the modified
valuations $\hat V_b$ are satisfied.
Under each of Valuations~\ref{additivity}--\ref{kdemand} in Section~\ref{sec:value}, the corresponding
valuation structure satisfies the gross substitutes property by
Proposition~\ref{pro:gs}. The modified-valuation equilibrium conditions are also verified in all four cases.
\section{Conclusions}

We demonstrate that combining firm and risk-aware flexible capacity assessment
with value-based allocation can substantially increase the hosting capacity of
existing transmission networks for large loads. By allowing eligible large loads
to accept limited interruption risk, the proposed framework can unlock latent
capacity without immediate network expansion, while the resulting firm and
flexible capacity products can be allocated among competing applicants through
the SAA. For radial networks, we show that the demand-normalized capacity
assessment does not sacrifice aggregate admitted capacity, while the meshed benchmark
shows that the resulting solutions remain close to the corresponding
aggregate-capacity maxima. Furthermore, for four
practically relevant valuation structures, the auction supports tractable bidding
and provides competitive-equilibrium guarantees. The framework therefore
provides a practical and theoretically grounded approach to utilizing scarce
transmission capacity for growing large-load interconnection. Future work will
incorporate generation costs, network-expansion costs, and the associated
system-level economic trade-offs.

\bibliographystyle{IEEEtran}
\bibliography{ref}
\section*{Acknowledgment}
The authors thank Prof.~Junjie Qin (Purdue University) and Dr.~Zhenyu Fan (PJM)
for their insightful comments.
\section*{Appendix}
\setcounter{subsection}{0}
\renewcommand{\thesubsection}{\Alph{subsection}}
\subsection{Definition of the Path Matrix}\label{definition}
For the capacity results, consider a radial network rooted at reference bus 1. With branches oriented away from the root, its path matrix $\boldsymbol A\in\mathbb R^{|J|\times|N|}$ has entries
\begin{equation}\label{eq:pathmatrix}
A_{e,i}=\begin{cases}
1,&e\text{ lies on the path from bus 1 to bus }i,\\
0,&\text{otherwise}.
\end{cases}
\end{equation}
Under the DC power flow convention for withdrawals, the shift-factor matrix is $\boldsymbol S=\boldsymbol A$ \cite{low2022power}. The downstream sets $D_e:=\{i:A_{e,i}=1\}$ are laminar: any two are disjoint or one contains the other. Let $I=N\setminus\Omega$ and fix $c_i=0$ for $i\in\Omega$ throughout.

\subsection{Proof of Proposition~\ref{merge}}\label{proofth1}
To handle both capacity assessments, consider the polyhedron
\begin{equation}\label{eq:laminarpoly}
\begin{split}
\mathcal P_I=\{\boldsymbol c:\;& c_i=0\ (i\in\Omega),\quad
\boldsymbol a\le\boldsymbol c\le\overline{\boldsymbol c},\\
&\underline{\boldsymbol\beta}\le\boldsymbol A\boldsymbol c
\le\overline{\boldsymbol\beta}\},
\end{split}
\end{equation}
where $\boldsymbol a\ge\boldsymbol0$ is a fixed lower bound. For the firm model, $\boldsymbol a=\boldsymbol0$, $\overline{\boldsymbol c}=\overline{\boldsymbol q}-\overline{\boldsymbol l}$, and the line bounds are the constants in~\eqref{ro_reform}. Assume that the model is feasible and the relevant optimum exists.

\begin{lemma}\label{lem:augmentation}
If $\boldsymbol c,\widehat{\boldsymbol c}\in\mathcal P_I$ and
$\sum_{i\in I}\widehat c_i>\sum_{i\in I}c_i$, then there exist $j\in I$ and $\eta>0$ such that $\boldsymbol c+\eta\boldsymbol e_j\in\mathcal P_I$, where $\boldsymbol e_j$ is the $j$th unit vector.
\end{lemma}
\noindent\emph{Proof.}
Let $\boldsymbol\Delta=\widehat{\boldsymbol c}-\boldsymbol c$ and $I_+=\{i\in I:\Delta_i>0\}$. For any branch $e$ whose upper constraint is tight at $\boldsymbol c$, feasibility of $\widehat{\boldsymbol c}$ gives
\begin{equation}\label{eq:tightsum}
\sum_{i\in D_e}\Delta_i\le0.
\end{equation}
Suppose every $i\in I_+$ has a tight branch on its root path. For each such $i$, select its tight branch closest to the root. The distinct selected branches have pairwise disjoint downstream sets: if one were an ancestor of another, the latter would not be the rootmost choice on that path. These sets cover $I_+$, so summing~\eqref{eq:tightsum} gives a nonpositive sum of $\Delta_i$ on their union. Outside the union, there are no positive components, contradicting $\sum_{i\in I}\Delta_i>0$.

Consequently, some $j\in I_+$ has no tight upper branch on its path $E_j$. Choose
\begin{equation}\label{eq:augmentationstep}
0<\eta\le\min\left\{\widehat c_j-c_j,
\min_{e\in E_j}\big[\overline\beta_e-(\boldsymbol A\boldsymbol c)_e\big]\right\},
\end{equation}
with the inner minimum equal to $+\infty$ for an empty path. Only coordinate $j$ increases, and $c_j+\eta\le\widehat c_j\le\overline c_j$; all other nodal upper bounds are unchanged. The path slacks preserve every line upper bound. Since $\boldsymbol A\ge0$, increasing a coordinate preserves the line lower bounds and the fixed floor $\boldsymbol a$. Thus $\boldsymbol c+\eta\boldsymbol e_j\in\mathcal P_I$.\hfill$\square$

\noindent\emph{Proof of Proposition~\ref{merge}.}
Write $F(\boldsymbol c)=\sum_{i\in I}((d_i-c_i)/d_i)^2$. If an optimizer $\boldsymbol c^*$ of the firm model did not maximize total capacity, Lemma~\ref{lem:augmentation} would provide a feasible increase in one coordinate $j$. By the proposition's scarcity assumption, $c_j^*<d_j$. Decrease $\eta$, if necessary, so that $0<\eta\le d_j-c_j^*$. Then
\begin{equation}
F(\boldsymbol c^*+\eta\boldsymbol e_j)-F(\boldsymbol c^*)
=\frac{\eta^2-2\eta(d_j-c_j^*)}{d_j^2}<0,
\end{equation}
contradicting optimality. Hence the same solution maximizes total firm capacity.\hfill$\square$

\subsection{Proof of Proposition~\ref{pro2}}\label{app:flexproof}
Translation equivariance of CVaR converts~\eqref{cvar_1}--\eqref{cvar_3} into
\begin{subequations}\label{eq:cvar-affine}
\begin{align}
\boldsymbol{c^r}&\le\overline{\boldsymbol q}-\CVaR_\alpha[\boldsymbol l],\\
\boldsymbol S\boldsymbol{c^r}&\le\overline{\boldsymbol b}-\CVaR_\alpha[\boldsymbol S\boldsymbol l],\\
\boldsymbol S\boldsymbol{c^r}&\ge\underline{\boldsymbol b}+\CVaR_\alpha[-\boldsymbol S\boldsymbol l].
\end{align}
\end{subequations}
For a radial network, this is~\eqref{eq:laminarpoly} with $\boldsymbol a=\boldsymbol{c^f}$ and adjusted constant upper and lower bounds. The augmentation lemma still holds with this nonzero floor. The proof of Proposition~\ref{merge} therefore applies under $c_i^{r\star}<d_i$. Because $\boldsymbol{c^f}$ is fixed, maximizing $\sum_i c_i^r$ and maximizing $\sum_i(c_i^r-c_i^f)$ are equivalent.\hfill$\square$

For completeness, the preserved-baseline model is feasible whenever the firm model is feasible, and the distribution used for CVaR is supported on the firm uncertainty box. Each firm constraint bounds its corresponding random quantity almost surely. Monotonicity of CVaR then gives the same bound on its CVaR, so $\boldsymbol{c^r}=\boldsymbol{c^f}$ satisfies~\eqref{cvar}. This statement also holds for an empirical distribution whose scenarios lie in that box.

\subsection{Scenario Approximation of the Flexible-Capacity Assessment}\label{app:reform}
For an integrable scalar random quantity $X$, the upper-tail CVaR is \cite{rockafellar2002conditional}
\begin{equation}\label{eq:cvar-definition}
\CVaR_\alpha[X]=\min_{\zeta\in\mathbb R}
\left\{\zeta+\frac{1}{1-\alpha}\mathbb E[(X-\zeta)^+]\right\}.
\end{equation}
Take $N_s$ background scenarios $\boldsymbol l[s]$, $s\in\mathcal S:=\{1,\ldots,N_s\}$, with equal empirical weights. Define the residuals
\begin{align*}
\boldsymbol g^{\rm cap}_s&:=\boldsymbol{c^r}+\boldsymbol l[s]-\overline{\boldsymbol q},\\
\boldsymbol g^{\rm up}_s&:=\boldsymbol S(\boldsymbol{c^r}+\boldsymbol l[s])-\overline{\boldsymbol b},\\
\boldsymbol g^{\rm low}_s&:=\underline{\boldsymbol b}-\boldsymbol S(\boldsymbol{c^r}+\boldsymbol l[s]).
\end{align*}
The empirical counterpart of~\eqref{cvar} is
\begin{subequations}\label{eq:empirical-cvar}
\begin{align}
\min_{\boldsymbol{c^r},\boldsymbol\zeta,\boldsymbol z}\quad&
\sum_{i\in I}\left(\frac{d_i-c_i^r}{d_i}\right)^2\\
\text{s.t.}\quad&\boldsymbol{c^r}\ge\boldsymbol{c^f},\quad c_i^r=0\quad(i\in\Omega),\\
&\boldsymbol\zeta^a+\frac{1}{(1-\alpha)N_s}
\sum_{s\in\mathcal S}\boldsymbol z_s^a\le\boldsymbol0,
\quad a\in\mathcal A,\\
&\boldsymbol z_s^a\ge\boldsymbol g_s^a-\boldsymbol\zeta^a,\quad
\boldsymbol z_s^a\ge\boldsymbol0,
\quad a\in\mathcal A,\ s\in\mathcal S,
\end{align}
\end{subequations}
where $\mathcal A=\{\mathrm{cap},\mathrm{up},\mathrm{low}\}$ and each $\boldsymbol\zeta^a$ is unrestricted. Equivalently, the empirical CVaR constants in~\eqref{eq:cvar-affine} can be evaluated before solving for capacity. This is an empirical implementation, not a distribution-free out-of-sample certificate. Here, \emph{SAA} always denotes the \emph{simultaneous ascending auction}, not a sample-average approximation method.

\subsection{Gross Substitutes for the Four Valuation Classes}\label{proofprogs}
We give a direct verification of the valuation claim in Proposition~\ref{pro:gs}. Throughout, prices are nonnegative and ties are allowed.

\noindent\emph{Additive valuations.}
Under Valuation~\ref{additivity}, utility is $\sum_{j\in U}(v_j^b-p_j)$. Every demanded bundle includes the positive-surplus products, excludes negative-surplus products, and may contain zero-surplus products. If prices increase, any previously demanded product whose price is unchanged still has nonnegative surplus and can be retained in an optimal bundle. This proves gross substitutes.

\medskip\noindent\emph{Symmetric concave valuations.}
Suppress $b$ and set $n=|K|$ and $\Delta_k=f(k)-f(k-1)$, so $\Delta_1\ge\cdots\ge\Delta_n\ge0$. An optimal bundle of size $k$ can be chosen from the $k$ cheapest products. The quantity objective has increments $\Delta_k-p_{(k)}$, which are nonincreasing in $k$.

It suffices to handle a price increase of one product $i$ and then apply the argument successively. Let $U$ be any demanded bundle at $\boldsymbol p$. If $i\notin U$, the utility of $U$ is unchanged while competing utilities cannot rise, so $U$ remains demanded. Otherwise write $m=|U|$ and $T=U\setminus\{i\}$. The case $T=\emptyset$ imposes no retention requirement. For $m\ge2$, optimality implies that $U$ comprises the $m$ cheapest products, up to ties, and that
\[
\max_{j\in U}p_j\le\Delta_m,
\]
by comparing $U$ with the bundle obtained by deleting its most expensive product.

If $p_i'\le\Delta_m$, there are at least $m$ products priced at most $\Delta_m$. The unchanged products in $T$ can be included among the $m$ cheapest products, because outside products were no cheaper than them before the increase. Since $\Delta_m-p'_{(m)}\ge0$, there exists an optimal quantity of at least $m$, and an optimal cheapest-product bundle can retain $T$.

If $p_i'>\Delta_m$, the set $T$ can be chosen among the cheapest $m-1$ products. Moreover,
$p'_{(m-1)}\le\max_{j\in T}p_j\le\Delta_m\le\Delta_{m-1}$.
Thus, an optimal quantity can be chosen to be at least $m-1$, again retaining $T$ through tie-breaking. This establishes the single-price property, and successive increases establish gross substitutes for any $\boldsymbol p'\ge\boldsymbol p$.

\medskip
\noindent\emph{$K_b$-demand and unit-demand valuations.}
Let $k=K_b$ and define the net value of product $j$ at prices
$\boldsymbol p$ as
\[
s_j=v_j-p_j.
\]
Under Valuation~\ref{kdemand}, a bidder's utility maximization problem is
equivalent to
\begin{equation}\label{eq:capped-surplus}
\max_{\substack{A\subseteq K\\ |A|\le k}}
\sum_{j\in A}s_j .
\end{equation}
Hence, an optimal bundle can be obtained by selecting up to $k$ products
with the largest nonnegative net values $s_j$.

Now consider two price vectors $\boldsymbol p'\ge\boldsymbol p$, and let
$U\in D(\boldsymbol p)$ be a demanded bundle. Let $W\subseteq U$,
$|W|\le k$, denote the products in $U$ that determine its $K_b$-demand
value. Because $U$ is optimal, $W$ consists of products with the largest
nonnegative net values at $\boldsymbol p$. Moreover, any product in
$U\setminus W$ must have zero price; otherwise removing it would strictly
increase the bidder's utility.

Consider any product $j\in W$ whose price does not change, i.e.,
$p'_j=p_j$. Its net value remains unchanged:
\[
v_j-p'_j=v_j-p_j.
\]
By contrast, the net value of every product whose price increases can only
decrease. Therefore, an optimal set $W'$ of at most $k$ products at
$\boldsymbol p'$ can be chosen to retain every product in $W$ whose price
did not change.

Finally, if $j\in U\setminus W$ and $p'_j=p_j$, then
$p_j=p'_j=0$. Such a product can be added to $W'$ without reducing utility.
Thus there exists a demanded bundle $U'\in D(\boldsymbol p')$ such that
\[
\{j\in U:p'_j=p_j\}\subseteq U'.
\]
Therefore the $K_b$-demand valuation satisfies the gross substitutes
property. The unit-demand valuation is the special case $K_b=1$, so the same argument
establishes Valuation~\ref{unit}.
\hfill$\square$

\subsection{Demand Computation with Standing Commitments}\label{app:demand}
Under the SAA rule, a bidder that does not already hold product $j$ must bid one increment $\epsilon$ above the standing high bid. Thus, during a round, the bidder faces the personalized price
\begin{equation}\label{eq:personal-price}
\widetilde p^{\,b}_j=
\begin{cases}
p_j, & \text{if $b$ is the standing high bidder on $j$},\\
p_j+\epsilon, & \text{otherwise}.
\end{cases}
\end{equation}
The preceding algorithms apply to a price vector. We now justify their use with the personalized prices in~\eqref{eq:personal-price} and a required standing bundle $H_b$.

First, a demanded superset of $H_b$ exists throughout straightforward bidding. Initially $H_b=\emptyset$. Suppose a bidder selects a demanded bundle containing its current holdings and places minimum bids on its new products. After the round, its new holdings are a subset of that selected bundle. Their personalized prices have not changed, whereas its other personalized prices can only increase. Gross substitutes therefore yields a new demanded bundle retaining all new holdings. This proves the invariant by induction, and also shows why standing bids are not withdrawn.

For additive values, retain $H_b$ and select positive-surplus products outside it; this takes $O(n)$ time. For symmetric concave values, sort the outside products by personalized price and scan
\[
f_b(|H_b|+t)-\sum_{j\in H_b}p_j
-\sum_{\ell=1}^{t}\widetilde p_{(\ell)}^b,
\quad 0\le t\le n-|H_b|.
\]
Sorting costs $O(n\log n)$ and the scan costs $O(n)$.

For $K_b$-demand, all standing prices in this zero-start, positive-increment auction are strictly positive. Since $H_b$ is contained in a demanded bundle, its size is at most $K_b$; otherwise a positive-price product unused in the value-maximizing $K_b$-subset could be removed to improve utility. Likewise, a demanded extension can be chosen to have at most $K_b$ products. Retain $H_b$, sort the remaining singleton net values, and select up to $K_b-|H_b|$ nonnegative entries. This is $O(n\log n)$. For unit-demand, either retain the single standing product or scan for a best nonnegative singleton surplus if the bidder has no holdings, giving $O(n)$ time.

\subsection{Finite-Increment and Limiting Allocation Guarantees}\label{app:auctionproof}
We include the argument underlying the established SAA conclusions \cite{milgrom2000saa} to make their application to the present products explicit. Let $M=\max_b V_b(K)<\infty$. Under straightforward bidding, any product receiving a new bid must have a personalized price no larger than the bidder's marginal value for that product in its demanded bundle, and hence no larger than $M$. Each bidding round with a new bid increases at least one standing price by at least $\epsilon$. Bounded prices and finitely many products imply finite termination.

At termination, $H_b=U_b^*$ and no new bid is submitted, so the demand bundle selected by bidder $b$ under straightforward bidding is its final bundle. Consequently, for any $U\subseteq K$,
\begin{equation}\label{eq:terminal-utility}
\begin{split}
V_b(U_b^*)-\sum_{j\in U_b^*}p_j^*
\ge V_b(U)-\sum_{j\in U}p_j^*
-\epsilon|U\setminus U_b^*|.
\end{split}
\end{equation}
This is precisely utility maximization for~\eqref{modifiedvaluation}. Unsold products have zero prices, so market clearing~\eqref{ce2} holds as well.

Apply~\eqref{eq:terminal-utility} to any feasible comparison allocation $(U_1,\ldots,U_m)$ and sum over bidders. The final allocation collects all positive product prices, while a comparison allocation can collect no more. Also, disjointness implies $\sum_b|U_b\setminus U_b^*|\le n$. Rearranging yields~\eqref{efficient}.

Finally, consider a sequence $\epsilon\downarrow0$. All terminal price vectors lie in the compact box $[0,M]^n$, and there are finitely many allocations. Thus a subsequence has a constant allocation and convergent prices. Taking its limit in~\eqref{eq:terminal-utility} removes the increment penalty, yielding~\eqref{ce1}; market clearing also passes to the limit. Every accumulation point is therefore a competitive equilibrium for the original valuations. This establishes existence as well, without claiming uniqueness of the limiting prices or allocation.\hfill$\square$

\onecolumn
\subsection{Numerical Auction Data and Verification}\label{app:verification}
This appendix reports the inputs and complete outcomes of the four SAA experiments
corresponding to the four valuation classes stated in Section~\ref{section3}. The
quantities of the firm and flexible capacity products are taken directly from the
capacity solution reported in Table~\ref{tab:capacity}; no product quantity is
modified before the auction.

\emph{Products and valuation inputs.} Table~\ref{tab:app-products} lists the
$|K|=\numNprod$ capacity products, one firm and one incremental flexible product
at each of the three candidate locations. Table~\ref{tab:app-values} gives the
standalone values $v_k^b$ used by the additive, unit-demand, and $K_b$-demand
settings. These values are obtained from each product's capacity, a fixed
location-preference matrix, and a bidder-specific sensitivity to the risk attribute,
and are then normalized to a maximum of $100$. Table~\ref{tab:app-fvals} gives
the symmetric concave
valuations $f_b$; each is the sum of that bidder's $q$ largest standalone
values, which makes it monotone and discretely concave by construction and keeps
it on the same value scale. The bid increment is $\epsilon=\numEps$ throughout,
and $K_b=\numKb$ in the $K_b$-demand setting.

\begin{table}[H]
\centering\small

\caption{Capacity products entering the auction, taken directly from the
$r=\numRiskPct\%$ solution of Section~\ref{section4}.}\label{tab:app-products}
\setlength{\tabcolsep}{12pt}
\begin{tabular}{clrrr}
\toprule
Product & Type & Bus & Capacity (MW) & Risk\\\midrule
1 & firm & 118 & 191.4 & 0 \\
2 & firm & 218 & 321.5 & 0 \\
3 & firm & 318 & 259.2 & 0 \\
4 & flexible & 118 & 165.4 & 0.05 \\
5 & flexible & 218 & 74.2 & 0.05 \\
6 & flexible & 318 & 117.9 & 0.05 \\
\bottomrule
\end{tabular}
\end{table}

\begin{table}[H]
\centering\small

\caption{Product values $v_k^b$, normalised to a maximum
of $100$.}\label{tab:app-values}
\setlength{\tabcolsep}{10pt}
\begin{tabular}{crrrrrr}
\toprule
Bidder & $k={1}$ & $k={2}$ & $k={3}$ & $k={4}$ & $k={5}$ & $k={6}$\\\midrule
1 & 59.52 & 75.00 & 40.30 & 50.16 & 16.88 & 17.88 \\
2 & 35.71 & 100.00 & 60.45 & 29.32 & 21.93 & 26.13 \\
3 & 44.64 & 50.00 & 80.60 & 35.69 & 10.68 & 33.93 \\
4 & 53.57 & 80.00 & 56.42 & 41.67 & 16.62 & 23.11 \\
\bottomrule
\end{tabular}
\end{table}

\begin{table}[H]
\centering\small

\caption{Symmetric concave valuations $f_b(q)$, $q=1,\dots,6$, with
$f_b(0)=0$.}\label{tab:app-fvals}
\setlength{\tabcolsep}{10pt}
\begin{tabular}{crrrrrr}
\toprule
Bidder & $q={1}$ & $q={2}$ & $q={3}$ & $q={4}$ & $q={5}$ & $q={6}$\\\midrule
1 & 75.00 & 134.52 & 184.68 & 224.98 & 242.86 & 259.74 \\
2 & 100.00 & 160.45 & 196.16 & 225.48 & 251.62 & 273.55 \\
3 & 80.60 & 130.60 & 175.24 & 210.93 & 244.85 & 255.53 \\
4 & 80.00 & 136.42 & 189.99 & 231.66 & 254.76 & 271.38 \\
\bottomrule
\end{tabular}
\end{table}

\clearpage
\emph{Bidding trajectories.} Tables~\ref{tab:app-rounds-additive}--%
\ref{tab:app-rounds-kdemand} report the complete round-by-round histories of all
four settings in a common format. Column $B_b$
lists the products on which bidder $b$ submits a new bid in round $t$;
$\boldsymbol p^{(t)}$ and ``holders'' are the standing high bids and standing
high bidders at the close of the round. The last row is the terminating round,
in which no bidder submits a new bid. Ties within a bidder's demand correspondence are resolved in favor of the
smaller bundle, so a bidder bids on a further product only when doing so strictly
increases its surplus; ties across bidders for the same product are resolved using
an alternating-priority rule. Each bid shown was generated by the polynomial-time
demand routine of the corresponding valuation class in Section~\ref{section3};
the bundle enumeration below is used only for ex-post verification.
\begingroup
\fontsize{8.5}{10.2}\selectfont
\setlength{\tabcolsep}{5pt}
\renewcommand{\arraystretch}{1.04}
\setlength{\LTleft}{\fill}
\setlength{\LTright}{\fill}
\setlength{\LTcapwidth}{\textwidth}
\begin{longtable}{@{}rcccccc@{}}
\caption{Complete SAA trajectory: Additive valuations.}\label{tab:app-rounds-additive}\\
\toprule
$t$ & $B_1$ & $B_2$ & $B_3$ & $B_4$ & $\boldsymbol p^{(t)}$ & Holders \\
\midrule
\endfirsthead
\multicolumn{7}{c}{\tablename~\thetable{} (continued)}\\
\toprule
$t$ & $B_1$ & $B_2$ & $B_3$ & $B_4$ & $\boldsymbol p^{(t)}$ & Holders \\
\midrule
\endhead
\midrule
\multicolumn{7}{r}{\emph{Continued on the next page}}\\
\endfoot
\bottomrule
\endlastfoot
1 & 1,2,3,4,5,6 & 1,2,3,4,5,6 & 1,2,3,4,5,6 & 1,2,3,4,5,6 & $(1,1,1,1,1,1)$ & $(1,2,3,4,1,2)$ \\
2 & 2,3,4,6 & 1,3,4,5 & 1,2,4,5,6 & 1,2,3,5,6 & $(2,2,2,2,2,2)$ & $(3,4,1,2,3,4)$ \\
3 & 1,2,4,5,6 & 1,2,3,5,6 & 2,3,4,6 & 1,3,4,5 & $(3,3,3,3,3,3)$ & $(1,2,3,4,1,2)$ \\
4 & 2,3,4,6 & 1,3,4,5 & 1,2,4,5,6 & 1,2,3,5,6 & $(4,4,4,4,4,4)$ & $(3,4,1,2,3,4)$ \\
5 & 1,2,4,5,6 & 1,2,3,5,6 & 2,3,4,6 & 1,3,4,5 & $(5,5,5,5,5,5)$ & $(1,2,3,4,1,2)$ \\
6 & 2,3,4,6 & 1,3,4,5 & 1,2,4,5,6 & 1,2,3,5,6 & $(6,6,6,6,6,6)$ & $(3,4,1,2,3,4)$ \\
7 & 1,2,4,5,6 & 1,2,3,5,6 & 2,3,4,6 & 1,3,4,5 & $(7,7,7,7,7,7)$ & $(1,2,3,4,1,2)$ \\
8 & 2,3,4,6 & 1,3,4,5 & 1,2,4,5,6 & 1,2,3,5,6 & $(8,8,8,8,8,8)$ & $(3,4,1,2,3,4)$ \\
9 & 1,2,4,5,6 & 1,2,3,5,6 & 2,3,4,6 & 1,3,4,5 & $(9,9,9,9,9,9)$ & $(1,2,3,4,1,2)$ \\
10 & 2,3,4,6 & 1,3,4,5 & 1,2,4,5,6 & 1,2,3,5,6 & $(10,10,10,10,10,10)$ & $(3,4,1,2,3,4)$ \\
11 & 1,2,4,5,6 & 1,2,3,5,6 & 2,3,4,6 & 1,3,4,5 & $(11,11,11,11,11,11)$ & $(1,2,3,4,1,2)$ \\
12 & 2,3,4,6 & 1,3,4,5 & 1,2,4,6 & 1,2,3,5,6 & $(12,12,12,12,12,12)$ & $(3,4,1,2,4,1)$ \\
13 & 1,2,4,5 & 1,2,3,5,6 & 2,3,4,6 & 1,3,4,6 & $(13,13,13,13,13,13)$ & $(2,3,4,1,2,3)$ \\
14 & 1,2,3,5,6 & 2,3,4,6 & 1,3,4 & 1,2,4,5,6 & $(14,14,14,14,14,14)$ & $(4,1,2,3,4,1)$ \\
15 & 1,3,4,5 & 1,2,4,5,6 & 1,2,3,6 & 2,3,4,6 & $(15,15,15,15,15,15)$ & $(2,3,4,1,2,3)$ \\
16 & 1,2,3,5,6 & 2,3,4,6 & 1,3,4 & 1,2,4,5,6 & $(16,16,16,16,16,16)$ & $(4,1,2,3,4,1)$ \\
17 & 1,3,4 & 1,2,4,5,6 & 1,2,3,6 & 2,3,4,6 & $(17,17,17,17,17,17)$ & $(2,3,4,1,2,2)$ \\
18 & 1,2,3 & 2,3,4 & 1,3,4,6 & 1,2,4,6 & $(18,18,18,18,17,18)$ & $(3,4,1,2,2,3)$ \\
19 & 1,2,4 & 1,2,3,6 & 2,3,4 & 1,3,4,6 & $(19,19,19,19,17,19)$ & $(4,1,2,3,2,4)$ \\
20 & 1,3,4 & 1,2,4,6 & 1,2,3,6 & 2,3,4 & $(20,20,20,20,17,20)$ & $(1,2,3,4,2,2)$ \\
21 & 2,3,4 & 1,3,4 & 1,2,4,6 & 1,2,3,6 & $(21,21,21,21,17,21)$ & $(3,4,1,2,2,3)$ \\
22 & 1,2,4 & 1,2,3,6 & 2,3,4 & 1,3,4,6 & $(22,22,22,22,17,22)$ & $(4,1,2,3,2,4)$ \\
23 & 1,3,4 & 1,2,4,6 & 1,2,3,6 & 2,3,4 & $(23,23,23,23,17,23)$ & $(1,2,3,4,2,2)$ \\
24 & 2,3,4 & 1,3,4 & 1,2,4,6 & 1,2,3 & $(24,24,24,24,17,24)$ & $(3,4,1,2,2,3)$ \\
25 & 1,2,4 & 1,2,3,6 & 2,3,4 & 1,3,4 & $(25,25,25,25,17,25)$ & $(4,1,2,3,2,2)$ \\
26 & 1,3,4 & 1,2,4 & 1,2,3,6 & 2,3,4 & $(26,26,26,26,17,26)$ & $(1,2,3,4,2,3)$ \\
27 & 2,3,4 & 1,3,4 & 1,2,4 & 1,2,3 & $(27,27,27,27,17,26)$ & $(2,3,4,1,2,3)$ \\
28 & 1,2,3 & 2,3,4 & 1,3,4 & 1,2,4 & $(28,28,28,28,17,26)$ & $(3,4,1,2,2,3)$ \\
29 & 1,2,4 & 1,2,3 & 2,3,4 & 1,3,4 & $(29,29,29,29,17,26)$ & $(4,1,2,3,2,3)$ \\
30 & 1,3,4 & 1,2 & 1,2,3 & 2,3,4 & $(30,30,30,30,17,26)$ & $(1,2,3,4,2,3)$ \\
31 & 2,3,4 & 1,3 & 1,2,4 & 1,2,3 & $(31,31,31,31,17,26)$ & $(2,3,4,1,2,3)$ \\
32 & 1,2,3 & 2,3 & 1,3,4 & 1,2,4 & $(32,32,32,32,17,26)$ & $(3,4,1,3,2,3)$ \\
33 & 1,2,4 & 1,2,3 & 2,3 & 1,3,4 & $(33,33,33,33,17,26)$ & $(4,1,2,4,2,3)$ \\
34 & 1,3,4 & 1,2 & 1,2,3,4 & 2,3 & $(34,34,34,34,17,26)$ & $(1,2,3,1,2,3)$ \\
35 & 2,3 & 1,3 & 1,2,4 & 1,2,3,4 & $(35,35,35,35,17,26)$ & $(2,3,4,3,2,3)$ \\
36 & 1,2,3,4 & 2,3 & 1,3 & 1,2,4 & $(36,36,36,36,17,26)$ & $(4,1,2,4,2,3)$ \\
37 & 1,3,4 & 2 & 1,2,3 & 2,3 & $(37,37,37,37,17,26)$ & $(1,2,3,1,2,3)$ \\
38 & 2,3 & 3 & 1,2 & 1,2,3,4 & $(38,38,38,38,17,26)$ & $(4,1,2,4,2,3)$ \\
39 & 1,3,4 & 2 & 1,2,3 & 2,3 & $(39,39,39,39,17,26)$ & $(3,4,1,1,2,3)$ \\
40 & 1,2 & 2,3 & 2,3 & 1,3,4 & $(40,40,40,40,17,26)$ & $(4,1,2,4,2,3)$ \\
41 & 1,4 & 2 & 1,2,3 & 2,3 & $(41,41,41,41,17,26)$ & $(3,4,3,1,2,3)$ \\
42 & 1,2 & 2,3 & 2 & 1,3 & $(42,42,42,41,17,26)$ & $(4,1,2,1,2,3)$ \\
43 & 1 & 2 & 1,2,3 & 2,3 & $(43,43,43,41,17,26)$ & $(3,4,3,1,2,3)$ \\
44 & 1,2 & 2,3 & 2 & 1,3 & $(44,44,44,41,17,26)$ & $(4,1,2,1,2,3)$ \\
45 & 1 & 2 & 2,3 & 2,3 & $(45,45,45,41,17,26)$ & $(1,3,4,1,2,3)$ \\
46 & 2 & 2,3 & 3 & 1,2 & $(46,46,46,41,17,26)$ & $(4,1,2,1,2,3)$ \\
47 & 1 & 2 & 2,3 & 2,3 & $(47,47,47,41,17,26)$ & $(1,3,4,1,2,3)$ \\
48 & 2 & 2,3 & 3 & 1,2 & $(48,48,48,41,17,26)$ & $(4,1,2,1,2,3)$ \\
49 & 1 & 2 & 2,3 & 2,3 & $(49,49,49,41,17,26)$ & $(1,3,4,1,2,3)$ \\
50 & 2 & 2,3 & 3 & 1,2 & $(50,50,50,41,17,26)$ & $(4,1,2,1,2,3)$ \\
51 & 1 & 2 & 3 & 2,3 & $(51,51,51,41,17,26)$ & $(1,4,3,1,2,3)$ \\
52 & 2 & 2,3 & -- & 1,3 & $(52,52,52,41,17,26)$ & $(4,1,2,1,2,3)$ \\
53 & 1 & 2 & 3 & 2,3 & $(53,53,53,41,17,26)$ & $(1,4,3,1,2,3)$ \\
54 & 2 & 2,3 & -- & 3 & $(53,54,54,41,17,26)$ & $(1,1,2,1,2,3)$ \\
55 & -- & 2 & 3 & 2,3 & $(53,55,55,41,17,26)$ & $(1,4,3,1,2,3)$ \\
56 & 2 & 2,3 & -- & 3 & $(53,56,56,41,17,26)$ & $(1,1,2,1,2,3)$ \\
57 & -- & 2 & 3 & 2 & $(53,57,57,41,17,26)$ & $(1,4,3,1,2,3)$ \\
58 & 2 & 2,3 & -- & -- & $(53,58,58,41,17,26)$ & $(1,1,2,1,2,3)$ \\
59 & -- & 2 & 3 & 2 & $(53,59,59,41,17,26)$ & $(1,2,3,1,2,3)$ \\
60 & 2 & 3 & -- & 2 & $(53,60,60,41,17,26)$ & $(1,4,2,1,2,3)$ \\
61 & 2 & 2 & 3 & -- & $(53,61,61,41,17,26)$ & $(1,1,3,1,2,3)$ \\
62 & -- & 2 & -- & 2 & $(53,62,61,41,17,26)$ & $(1,2,3,1,2,3)$ \\
63 & 2 & -- & -- & 2 & $(53,63,61,41,17,26)$ & $(1,4,3,1,2,3)$ \\
64 & 2 & 2 & -- & -- & $(53,64,61,41,17,26)$ & $(1,1,3,1,2,3)$ \\
65 & -- & 2 & -- & 2 & $(53,65,61,41,17,26)$ & $(1,2,3,1,2,3)$ \\
66 & 2 & -- & -- & 2 & $(53,66,61,41,17,26)$ & $(1,4,3,1,2,3)$ \\
67 & 2 & 2 & -- & -- & $(53,67,61,41,17,26)$ & $(1,1,3,1,2,3)$ \\
68 & -- & 2 & -- & 2 & $(53,68,61,41,17,26)$ & $(1,2,3,1,2,3)$ \\
69 & 2 & -- & -- & 2 & $(53,69,61,41,17,26)$ & $(1,4,3,1,2,3)$ \\
70 & 2 & 2 & -- & -- & $(53,70,61,41,17,26)$ & $(1,1,3,1,2,3)$ \\
71 & -- & 2 & -- & 2 & $(53,71,61,41,17,26)$ & $(1,2,3,1,2,3)$ \\
72 & 2 & -- & -- & 2 & $(53,72,61,41,17,26)$ & $(1,4,3,1,2,3)$ \\
73 & 2 & 2 & -- & -- & $(53,73,61,41,17,26)$ & $(1,1,3,1,2,3)$ \\
74 & -- & 2 & -- & 2 & $(53,74,61,41,17,26)$ & $(1,2,3,1,2,3)$ \\
75 & -- & -- & -- & 2 & $(53,75,61,41,17,26)$ & $(1,4,3,1,2,3)$ \\
76 & -- & 2 & -- & -- & $(53,76,61,41,17,26)$ & $(1,2,3,1,2,3)$ \\
77 & -- & -- & -- & 2 & $(53,77,61,41,17,26)$ & $(1,4,3,1,2,3)$ \\
78 & -- & 2 & -- & -- & $(53,78,61,41,17,26)$ & $(1,2,3,1,2,3)$ \\
79 & -- & -- & -- & 2 & $(53,79,61,41,17,26)$ & $(1,4,3,1,2,3)$ \\
80 & -- & 2 & -- & -- & $(53,80,61,41,17,26)$ & $(1,2,3,1,2,3)$ \\
81 & -- & -- & -- & -- & $(53,80,61,41,17,26)$ & $(1,2,3,1,2,3)$ \\
\end{longtable}
\endgroup
\begingroup
\fontsize{8.5}{10.2}\selectfont
\setlength{\tabcolsep}{5pt}
\renewcommand{\arraystretch}{1.04}
\setlength{\LTleft}{\fill}
\setlength{\LTright}{\fill}
\setlength{\LTcapwidth}{\textwidth}
\begin{longtable}{@{}rcccccc@{}}
\caption{Complete SAA trajectory: Symmetric concave valuations.}\label{tab:app-rounds-symmetric}\\
\toprule
$t$ & $B_1$ & $B_2$ & $B_3$ & $B_4$ & $\boldsymbol p^{(t)}$ & Holders \\
\midrule
\endfirsthead
\multicolumn{7}{c}{\tablename~\thetable{} (continued)}\\
\toprule
$t$ & $B_1$ & $B_2$ & $B_3$ & $B_4$ & $\boldsymbol p^{(t)}$ & Holders \\
\midrule
\endhead
\midrule
\multicolumn{7}{r}{\emph{Continued on the next page}}\\
\endfoot
\bottomrule
\endlastfoot
1 & 1,2,3,4,5,6 & 1,2,3,4,5,6 & 1,2,3,4,5,6 & 1,2,3,4,5,6 & $(1,1,1,1,1,1)$ & $(1,2,3,4,1,2)$ \\
2 & 2,3,4,6 & 1,3,4,5 & 1,2,4,5,6 & 1,2,3,5,6 & $(2,2,2,2,2,2)$ & $(3,4,1,2,3,4)$ \\
3 & 1,2,4,5,6 & 1,2,3,5,6 & 2,3,4,6 & 1,3,4,5 & $(3,3,3,3,3,3)$ & $(1,2,3,4,1,2)$ \\
4 & 2,3,4,6 & 1,3,4,5 & 1,2,4,5,6 & 1,2,3,5,6 & $(4,4,4,4,4,4)$ & $(3,4,1,2,3,4)$ \\
5 & 1,2,4,5,6 & 1,2,3,5,6 & 2,3,4,6 & 1,3,4,5 & $(5,5,5,5,5,5)$ & $(1,2,3,4,1,2)$ \\
6 & 2,3,4,6 & 1,3,4,5 & 1,2,4,5,6 & 1,2,3,5,6 & $(6,6,6,6,6,6)$ & $(3,4,1,2,3,4)$ \\
7 & 1,2,4,5,6 & 1,2,3,5,6 & 2,3,4,6 & 1,3,4,5 & $(7,7,7,7,7,7)$ & $(1,2,3,4,1,2)$ \\
8 & 2,3,4,6 & 1,3,4,5 & 1,2,4,5,6 & 1,2,3,5,6 & $(8,8,8,8,8,8)$ & $(3,4,1,2,3,4)$ \\
9 & 1,2,4,5,6 & 1,2,3,5,6 & 2,3,4,6 & 1,3,4,5 & $(9,9,9,9,9,9)$ & $(1,2,3,4,1,2)$ \\
10 & 2,3,4,6 & 1,3,4,5 & 1,2,4,5,6 & 1,2,3,5,6 & $(10,10,10,10,10,10)$ & $(3,4,1,2,3,4)$ \\
11 & 1,2,4,5,6 & 1,2,3,5,6 & 2,3,4 & 1,3,4,5 & $(11,11,11,11,11,11)$ & $(1,2,3,4,1,2)$ \\
12 & 2,3,4,6 & 1,3,4,5 & 1,2,4,5 & 1,2,3,5,6 & $(12,12,12,12,12,12)$ & $(3,4,1,2,3,4)$ \\
13 & 1,2,4,5,6 & 1,2,3,5,6 & 2,3,4 & 1,3,4,5 & $(13,13,13,13,13,13)$ & $(1,2,3,4,1,2)$ \\
14 & 2,3,4,6 & 1,3,4,5 & 1,2,4,5 & 1,2,3,5,6 & $(14,14,14,14,14,14)$ & $(3,4,1,2,3,4)$ \\
15 & 1,2,4,5,6 & 1,2,3,5,6 & 2,3,4 & 1,3,4,5 & $(15,15,15,15,15,15)$ & $(1,2,3,4,1,2)$ \\
16 & 2,3,4,6 & 1,3,4,5 & 1,2,4,5 & 1,2,3,5,6 & $(16,16,16,16,16,16)$ & $(3,4,1,2,3,4)$ \\
17 & 1,2,4,5 & 1,2,3,5,6 & 2,3,4 & 1,3,4 & $(17,17,17,17,17,17)$ & $(1,2,3,4,1,2)$ \\
18 & 2,3 & 1,3,4,5 & 1,2,4,5 & 1,2,3,5 & $(18,18,18,18,18,17)$ & $(2,3,4,2,3,2)$ \\
19 & 1,2,3,6 & 2,3,5 & 1,3,6 & 1,2,4,6 & $(19,19,19,19,19,18)$ & $(4,1,2,4,2,3)$ \\
20 & 1,3,6 & 1,2,4,6 & 1,2,3,4 & 2,3,6 & $(20,20,20,20,19,19)$ & $(1,2,3,2,2,4)$ \\
21 & 2,5,6 & 1,3,6 & 1,2,5,6 & 1,2,3,5 & $(21,21,21,20,20,20)$ & $(2,3,4,2,1,2)$ \\
22 & 1,4,6 & 2,5 & 1,4,5,6 & 1,4,5,6 & $(22,22,21,21,21,21)$ & $(3,2,4,4,2,3)$ \\
23 & 3,4,5,6 & 3,4,6 & 3,4,5 & 1,5,6 & $(23,22,22,22,22,22)$ & $(4,2,1,2,3,4)$ \\
24 & 2,4,5 & 3,5,6 & 2,3,4,6 & 2,3,4 & $(23,23,23,23,23,23)$ & $(4,1,2,3,1,2)$ \\
25 & 1,3 & 1,2,4 & 1,2,3,5 & 2,3,4 & $(24,24,24,24,24,23)$ & $(3,4,1,2,3,2)$ \\
26 & 1,2,6 & 1,2,3 & 2,3,6 & 1,3,6 & $(25,25,25,24,24,24)$ & $(4,1,2,2,3,3)$ \\
27 & 4,5,6 & 1,5,6 & 1,2,4 & 4,5,6 & $(26,26,25,25,25,25)$ & $(2,3,2,3,4,1)$ \\
28 & 3,4,5 & 4,5,6 & 3,5,6 & 3,4,6 & $(26,26,26,26,26,26)$ & $(2,3,3,4,1,2)$ \\
29 & 1,2,3 & 2,3 & 1,4,5 & 1,2,3 & $(27,27,27,27,27,26)$ & $(3,4,1,3,3,2)$ \\
30 & 1,2,6 & 1,2,3 & 2,6 & 1,3,6 & $(28,28,28,27,27,27)$ & $(2,3,4,3,3,1)$ \\
31 & 1,4,5 & 4,5,6 & 1,6 & 4,5,6 & $(29,28,28,28,28,28)$ & $(3,3,4,4,1,2)$ \\
32 & 2,3,4 & 2,3,4 & 3,4,5 & 2,5 & $(29,29,29,29,29,28)$ & $(3,4,1,2,3,2)$ \\
33 & 1,2,6 & 1 & 2,3,6 & 1,3,6 & $(30,30,30,29,29,29)$ & $(4,1,3,2,3,4)$ \\
34 & 4,5,6 & 5,6 & 1,4,6 & 4,5 & $(31,30,30,30,30,30)$ & $(3,1,3,1,2,3)$ \\
35 & 3,5 & 2,3 & 2,4 & 2,3,4,5 & $(31,31,31,31,31,30)$ & $(3,4,1,3,4,3)$ \\
36 & 1,2,6 & 1,2,6 & 2,3 & 1,6 & $(32,32,32,31,31,31)$ & $(1,2,3,3,4,4)$ \\
37 & 4,5,6 & 4,5 & 1,5,6 & 1,4 & $(33,32,32,32,32,32)$ & $(3,2,3,4,1,3)$ \\
38 & 2,3,4 & 3,4 & 2,4 & 2,3,5 & $(33,33,33,33,33,32)$ & $(3,4,1,2,4,3)$ \\
39 & 1,2,6 & 1,6 & 2,3 & 1,6 & $(34,34,34,33,33,33)$ & $(4,1,3,2,4,2)$ \\
40 & 4,5,6 & 5 & 4,5,6 & 4,6 & $(34,34,34,34,34,34)$ & $(4,1,3,3,1,3)$ \\
41 & 1,3 & 1,2,3 & 1 & 2,3,4 & $(35,35,35,35,34,34)$ & $(1,2,4,4,1,3)$ \\
42 & 2,6 & 5,6 & 1,5 & 5,6 & $(36,36,35,35,35,35)$ & $(3,1,4,4,2,4)$ \\
43 & 3,4,5 & 3 & 3,4 & 5 & $(36,36,36,36,36,35)$ & $(3,1,1,3,4,4)$ \\
44 & 1,6 & 1,6 & 6 & 1,2 & $(37,37,36,36,36,36)$ & $(1,4,1,3,4,2)$ \\
45 & 4,5 & 3 & 3,5 & 3,4 & $(37,37,37,37,37,36)$ & $(1,4,3,4,1,2)$ \\
46 & 2,6 & 1 & 1,6 & 1,6 & $(38,38,37,37,37,37)$ & $(2,1,3,4,1,3)$ \\
47 & 3,4 & 3 & 4 & 3,5,6 & $(38,38,38,38,38,38)$ & $(2,1,4,1,4,4)$ \\
48 & 1,3 & 2 & 1,2,3 & 1 & $(39,39,39,38,38,38)$ & $(3,2,3,1,4,4)$ \\
49 & 1,5,6 & 4 & 4 & 1,4 & $(40,39,39,39,39,39)$ & $(4,2,3,2,1,1)$ \\
50 & 2,3 & -- & 2,4 & 2,3,4 & $(40,40,40,40,39,39)$ & $(4,3,4,3,1,1)$ \\
51 & 1 & 5,6 & 5 & 5,6 & $(41,40,40,40,40,40)$ & $(1,3,4,3,4,2)$ \\
52 & 2,3 & 2 & 3 & 2,4 & $(41,41,41,41,40,40)$ & $(1,4,1,4,4,2)$ \\
53 & 5 & 5 & 1,5,6 & 6 & $(42,41,41,41,41,41)$ & $(3,4,1,4,2,3)$ \\
54 & 2,4 & 2 & 2 & 3 & $(42,42,42,42,41,41)$ & $(3,1,4,1,2,3)$ \\
55 & 5 & 6 & 5 & 5,6 & $(42,42,42,42,42,42)$ & $(3,1,4,1,3,4)$ \\
56 & 1 & 1,2 & 2 & 1 & $(43,43,42,42,42,42)$ & $(1,2,4,1,3,4)$ \\
57 & 3 & 3 & 3,4 & 4 & $(43,43,43,43,42,42)$ & $(1,2,3,4,3,4)$ \\
58 & 5,6 & 5 & 6 & 5 & $(43,43,43,43,43,43)$ & $(1,2,3,4,1,3)$ \\
59 & 2 & 1 & 1 & 1,2 & $(44,44,43,43,43,43)$ & $(4,1,3,4,1,3)$ \\
60 & 3 & 3,4 & 4 & 3 & $(44,44,44,44,43,43)$ & $(4,1,2,3,1,3)$ \\
61 & 6 & 5 & 5 & 5,6 & $(44,44,44,44,44,44)$ & $(4,1,2,3,4,1)$ \\
62 & 1 & 1 & 1 & 2 & $(45,45,44,44,44,44)$ & $(2,4,2,3,4,1)$ \\
63 & 3,4 & -- & 3 & 3 & $(45,45,45,45,44,44)$ & $(2,4,3,1,4,1)$ \\
64 & 5 & 5 & 5 & 6 & $(45,45,45,45,45,45)$ & $(2,4,3,1,1,4)$ \\
65 & 1 & 2 & 1 & 1 & $(46,46,45,45,45,45)$ & $(3,2,3,1,1,4)$ \\
66 & 3 & 3 & -- & 3,4 & $(46,46,46,46,45,45)$ & $(3,2,4,4,1,4)$ \\
67 & 1,6 & 5 & 5 & -- & $(47,46,46,46,46,46)$ & $(1,2,4,4,2,1)$ \\
68 & 2 & -- & 2,3 & 2 & $(47,47,47,46,46,46)$ & $(1,3,3,4,2,1)$ \\
69 & 4 & 4 & -- & 5,6 & $(47,47,47,47,47,47)$ & $(1,3,3,1,4,4)$ \\
70 & 2 & 1,2 & -- & 1 & $(48,48,47,47,47,47)$ & $(2,1,3,1,4,4)$ \\
71 & 3 & 3 & 4 & 3 & $(48,48,48,48,47,47)$ & $(2,1,2,3,4,4)$ \\
72 & 5,6 & -- & 5 & 1 & $(49,48,48,48,48,48)$ & $(4,1,2,3,3,1)$ \\
73 & 3 & 2 & -- & 2,3 & $(49,49,49,48,48,48)$ & $(4,4,1,3,3,1)$ \\
74 & 4 & 4,5 & -- & 4 & $(49,49,49,49,49,48)$ & $(4,4,1,2,2,1)$ \\
75 & 1 & -- & 6 & 6 & $(50,49,49,49,49,49)$ & $(1,4,1,2,2,3)$ \\
76 & 2 & -- & -- & 3,4 & $(50,50,50,50,49,49)$ & $(1,1,4,4,2,3)$ \\
77 & 5 & 6 & -- & 5 & $(50,50,50,50,50,50)$ & $(1,1,4,4,4,2)$ \\
78 & -- & 1 & 1 & -- & $(51,50,50,50,50,50)$ & $(2,1,4,4,4,2)$ \\
79 & 3 & -- & 2 & -- & $(51,51,51,50,50,50)$ & $(2,3,1,4,4,2)$ \\
80 & 4 & -- & -- & 6 & $(51,51,51,51,50,51)$ & $(2,3,1,1,4,4)$ \\
81 & -- & 5 & -- & 1 & $(52,51,51,51,51,51)$ & $(4,3,1,1,2,4)$ \\
82 & -- & 2 & -- & 2 & $(52,52,51,51,51,51)$ & $(4,4,1,1,2,4)$ \\
83 & -- & 3 & 3 & -- & $(52,52,52,51,51,51)$ & $(4,4,2,1,2,4)$ \\
84 & 5 & -- & 4 & -- & $(52,52,52,52,52,51)$ & $(4,4,2,3,1,4)$ \\
85 & 6 & 6 & -- & -- & $(52,52,52,52,52,52)$ & $(4,4,2,3,1,1)$ \\
86 & -- & 1 & -- & 3 & $(53,52,53,52,52,52)$ & $(2,4,4,3,1,1)$ \\
87 & -- & 2 & -- & 4 & $(53,53,53,53,52,52)$ & $(2,2,4,4,1,1)$ \\
88 & -- & -- & 5 & 5 & $(53,53,53,53,53,52)$ & $(2,2,4,4,3,1)$ \\
89 & 1 & -- & -- & 6 & $(54,53,53,53,53,53)$ & $(1,2,4,4,3,4)$ \\
90 & 2 & 3 & -- & -- & $(54,54,54,53,53,53)$ & $(1,1,2,4,3,4)$ \\
91 & -- & 4 & -- & -- & $(54,54,54,54,53,53)$ & $(1,1,2,2,3,4)$ \\
92 & -- & -- & -- & 5 & $(54,54,54,54,54,53)$ & $(1,1,2,2,4,4)$ \\
93 & -- & -- & 6 & -- & $(54,54,54,54,54,54)$ & $(1,1,2,2,4,3)$ \\
94 & -- & -- & -- & 1 & $(55,54,54,54,54,54)$ & $(4,1,2,2,4,3)$ \\
95 & 3 & -- & -- & -- & $(55,54,55,54,54,54)$ & $(4,1,1,2,4,3)$ \\
96 & -- & 2 & -- & -- & $(55,55,55,54,54,54)$ & $(4,2,1,2,4,3)$ \\
97 & 4 & -- & -- & -- & $(55,55,55,55,54,54)$ & $(4,2,1,1,4,3)$ \\
98 & -- & 5 & -- & -- & $(55,55,55,55,55,54)$ & $(4,2,1,1,2,3)$ \\
99 & -- & -- & -- & 6 & $(55,55,55,55,55,55)$ & $(4,2,1,1,2,4)$ \\
100 & -- & -- & 1 & -- & $(56,55,55,55,55,55)$ & $(3,2,1,1,2,4)$ \\
101 & -- & -- & -- & 2 & $(56,56,55,55,55,55)$ & $(3,4,1,1,2,4)$ \\
102 & -- & 3 & -- & -- & $(56,56,56,55,55,55)$ & $(3,4,2,1,2,4)$ \\
103 & 5 & -- & -- & -- & $(56,56,56,55,56,55)$ & $(3,4,2,1,1,4)$ \\
104 & -- & 4 & -- & -- & $(56,56,56,56,56,55)$ & $(3,4,2,2,1,4)$ \\
105 & 6 & -- & -- & -- & $(56,56,56,56,56,56)$ & $(3,4,2,2,1,1)$ \\
106 & -- & -- & -- & -- & $(56,56,56,56,56,56)$ & $(3,4,2,2,1,1)$ \\
\end{longtable}
\endgroup
\begingroup
\fontsize{8.5}{10.2}\selectfont
\setlength{\tabcolsep}{5pt}
\renewcommand{\arraystretch}{1.04}
\setlength{\LTleft}{\fill}
\setlength{\LTright}{\fill}
\setlength{\LTcapwidth}{\textwidth}
\begin{longtable}{@{}rcccccc@{}}
\caption{Complete SAA trajectory: Unit-demand valuations.}\label{tab:app-rounds-unit}\\
\toprule
$t$ & $B_1$ & $B_2$ & $B_3$ & $B_4$ & $\boldsymbol p^{(t)}$ & Holders \\
\midrule
\endfirsthead
\multicolumn{7}{c}{\tablename~\thetable{} (continued)}\\
\toprule
$t$ & $B_1$ & $B_2$ & $B_3$ & $B_4$ & $\boldsymbol p^{(t)}$ & Holders \\
\midrule
\endhead
\midrule
\multicolumn{7}{r}{\emph{Continued on the next page}}\\
\endfoot
\bottomrule
\endlastfoot
1 & 2 & 2 & 3 & 2 & $(0,1,1,0,0,0)$ & $(\text{--},1,3,\text{--},\text{--},\text{--})$ \\
2 & -- & 2 & -- & 2 & $(0,2,1,0,0,0)$ & $(\text{--},2,3,\text{--},\text{--},\text{--})$ \\
3 & 2 & -- & -- & 2 & $(0,3,1,0,0,0)$ & $(\text{--},4,3,\text{--},\text{--},\text{--})$ \\
4 & 2 & 2 & -- & -- & $(0,4,1,0,0,0)$ & $(\text{--},1,3,\text{--},\text{--},\text{--})$ \\
5 & -- & 2 & -- & 2 & $(0,5,1,0,0,0)$ & $(\text{--},2,3,\text{--},\text{--},\text{--})$ \\
6 & 2 & -- & -- & 2 & $(0,6,1,0,0,0)$ & $(\text{--},4,3,\text{--},\text{--},\text{--})$ \\
7 & 2 & 2 & -- & -- & $(0,7,1,0,0,0)$ & $(\text{--},1,3,\text{--},\text{--},\text{--})$ \\
8 & -- & 2 & -- & 2 & $(0,8,1,0,0,0)$ & $(\text{--},2,3,\text{--},\text{--},\text{--})$ \\
9 & 2 & -- & -- & 2 & $(0,9,1,0,0,0)$ & $(\text{--},4,3,\text{--},\text{--},\text{--})$ \\
10 & 2 & 2 & -- & -- & $(0,10,1,0,0,0)$ & $(\text{--},1,3,\text{--},\text{--},\text{--})$ \\
11 & -- & 2 & -- & 2 & $(0,11,1,0,0,0)$ & $(\text{--},2,3,\text{--},\text{--},\text{--})$ \\
12 & 2 & -- & -- & 2 & $(0,12,1,0,0,0)$ & $(\text{--},4,3,\text{--},\text{--},\text{--})$ \\
13 & 2 & 2 & -- & -- & $(0,13,1,0,0,0)$ & $(\text{--},1,3,\text{--},\text{--},\text{--})$ \\
14 & -- & 2 & -- & 2 & $(0,14,1,0,0,0)$ & $(\text{--},2,3,\text{--},\text{--},\text{--})$ \\
15 & 2 & -- & -- & 2 & $(0,15,1,0,0,0)$ & $(\text{--},4,3,\text{--},\text{--},\text{--})$ \\
16 & 2 & 2 & -- & -- & $(0,16,1,0,0,0)$ & $(\text{--},1,3,\text{--},\text{--},\text{--})$ \\
17 & -- & 2 & -- & 2 & $(0,17,1,0,0,0)$ & $(\text{--},2,3,\text{--},\text{--},\text{--})$ \\
18 & 1 & -- & -- & 2 & $(1,18,1,0,0,0)$ & $(1,4,3,\text{--},\text{--},\text{--})$ \\
19 & -- & 2 & -- & -- & $(1,19,1,0,0,0)$ & $(1,2,3,\text{--},\text{--},\text{--})$ \\
20 & -- & -- & -- & 2 & $(1,20,1,0,0,0)$ & $(1,4,3,\text{--},\text{--},\text{--})$ \\
21 & -- & 2 & -- & -- & $(1,21,1,0,0,0)$ & $(1,2,3,\text{--},\text{--},\text{--})$ \\
22 & -- & -- & -- & 2 & $(1,22,1,0,0,0)$ & $(1,4,3,\text{--},\text{--},\text{--})$ \\
23 & -- & 2 & -- & -- & $(1,23,1,0,0,0)$ & $(1,2,3,\text{--},\text{--},\text{--})$ \\
24 & -- & -- & -- & 2 & $(1,24,1,0,0,0)$ & $(1,4,3,\text{--},\text{--},\text{--})$ \\
25 & -- & 2 & -- & -- & $(1,25,1,0,0,0)$ & $(1,2,3,\text{--},\text{--},\text{--})$ \\
26 & -- & -- & -- & 3 & $(1,25,2,0,0,0)$ & $(1,2,4,\text{--},\text{--},\text{--})$ \\
27 & -- & -- & 3 & -- & $(1,25,3,0,0,0)$ & $(1,2,3,\text{--},\text{--},\text{--})$ \\
28 & -- & -- & -- & 2 & $(1,26,3,0,0,0)$ & $(1,4,3,\text{--},\text{--},\text{--})$ \\
29 & -- & 2 & -- & -- & $(1,27,3,0,0,0)$ & $(1,2,3,\text{--},\text{--},\text{--})$ \\
30 & -- & -- & -- & 3 & $(1,27,4,0,0,0)$ & $(1,2,4,\text{--},\text{--},\text{--})$ \\
31 & -- & -- & 3 & -- & $(1,27,5,0,0,0)$ & $(1,2,3,\text{--},\text{--},\text{--})$ \\
32 & -- & -- & -- & 2 & $(1,28,5,0,0,0)$ & $(1,4,3,\text{--},\text{--},\text{--})$ \\
33 & -- & 2 & -- & -- & $(1,29,5,0,0,0)$ & $(1,2,3,\text{--},\text{--},\text{--})$ \\
34 & -- & -- & -- & 1 & $(2,29,5,0,0,0)$ & $(4,2,3,\text{--},\text{--},\text{--})$ \\
35 & 1 & -- & -- & -- & $(3,29,5,0,0,0)$ & $(1,2,3,\text{--},\text{--},\text{--})$ \\
36 & -- & -- & -- & 3 & $(3,29,6,0,0,0)$ & $(1,2,4,\text{--},\text{--},\text{--})$ \\
37 & -- & -- & 3 & -- & $(3,29,7,0,0,0)$ & $(1,2,3,\text{--},\text{--},\text{--})$ \\
38 & -- & -- & -- & 2 & $(3,30,7,0,0,0)$ & $(1,4,3,\text{--},\text{--},\text{--})$ \\
39 & -- & 2 & -- & -- & $(3,31,7,0,0,0)$ & $(1,2,3,\text{--},\text{--},\text{--})$ \\
40 & -- & -- & -- & 1 & $(4,31,7,0,0,0)$ & $(4,2,3,\text{--},\text{--},\text{--})$ \\
41 & 1 & -- & -- & -- & $(5,31,7,0,0,0)$ & $(1,2,3,\text{--},\text{--},\text{--})$ \\
42 & -- & -- & -- & 3 & $(5,31,8,0,0,0)$ & $(1,2,4,\text{--},\text{--},\text{--})$ \\
43 & -- & -- & 3 & -- & $(5,31,9,0,0,0)$ & $(1,2,3,\text{--},\text{--},\text{--})$ \\
44 & -- & -- & -- & 2 & $(5,32,9,0,0,0)$ & $(1,4,3,\text{--},\text{--},\text{--})$ \\
45 & -- & 2 & -- & -- & $(5,33,9,0,0,0)$ & $(1,2,3,\text{--},\text{--},\text{--})$ \\
46 & -- & -- & -- & 1 & $(6,33,9,0,0,0)$ & $(4,2,3,\text{--},\text{--},\text{--})$ \\
47 & 1 & -- & -- & -- & $(7,33,9,0,0,0)$ & $(1,2,3,\text{--},\text{--},\text{--})$ \\
48 & -- & -- & -- & 3 & $(7,33,10,0,0,0)$ & $(1,2,4,\text{--},\text{--},\text{--})$ \\
49 & -- & -- & 3 & -- & $(7,33,11,0,0,0)$ & $(1,2,3,\text{--},\text{--},\text{--})$ \\
50 & -- & -- & -- & 2 & $(7,34,11,0,0,0)$ & $(1,4,3,\text{--},\text{--},\text{--})$ \\
51 & -- & 2 & -- & -- & $(7,35,11,0,0,0)$ & $(1,2,3,\text{--},\text{--},\text{--})$ \\
52 & -- & -- & -- & 1 & $(8,35,11,0,0,0)$ & $(4,2,3,\text{--},\text{--},\text{--})$ \\
53 & 1 & -- & -- & -- & $(9,35,11,0,0,0)$ & $(1,2,3,\text{--},\text{--},\text{--})$ \\
54 & -- & -- & -- & 3 & $(9,35,12,0,0,0)$ & $(1,2,4,\text{--},\text{--},\text{--})$ \\
55 & -- & -- & 3 & -- & $(9,35,13,0,0,0)$ & $(1,2,3,\text{--},\text{--},\text{--})$ \\
56 & -- & -- & -- & 2 & $(9,36,13,0,0,0)$ & $(1,4,3,\text{--},\text{--},\text{--})$ \\
57 & -- & 2 & -- & -- & $(9,37,13,0,0,0)$ & $(1,2,3,\text{--},\text{--},\text{--})$ \\
58 & -- & -- & -- & 1 & $(10,37,13,0,0,0)$ & $(4,2,3,\text{--},\text{--},\text{--})$ \\
59 & 4 & -- & -- & -- & $(10,37,13,1,0,0)$ & $(4,2,3,1,\text{--},\text{--})$ \\
60 & -- & -- & -- & -- & $(10,37,13,1,0,0)$ & $(4,2,3,1,\text{--},\text{--})$ \\
\end{longtable}
\endgroup
\begingroup
\fontsize{8.5}{10.2}\selectfont
\setlength{\tabcolsep}{5pt}
\renewcommand{\arraystretch}{1.04}
\setlength{\LTleft}{\fill}
\setlength{\LTright}{\fill}
\setlength{\LTcapwidth}{\textwidth}
\begin{longtable}{@{}rcccccc@{}}
\caption{Complete SAA trajectory: $K_b$-demand valuations.}\label{tab:app-rounds-kdemand}\\
\toprule
$t$ & $B_1$ & $B_2$ & $B_3$ & $B_4$ & $\boldsymbol p^{(t)}$ & Holders \\
\midrule
\endfirsthead
\multicolumn{7}{c}{\tablename~\thetable{} (continued)}\\
\toprule
$t$ & $B_1$ & $B_2$ & $B_3$ & $B_4$ & $\boldsymbol p^{(t)}$ & Holders \\
\midrule
\endhead
\midrule
\multicolumn{7}{r}{\emph{Continued on the next page}}\\
\endfoot
\bottomrule
\endlastfoot
1 & 1,2 & 2,3 & 2,3 & 2,3 & $(1,1,1,0,0,0)$ & $(1,1,2,\text{--},\text{--},\text{--})$ \\
2 & -- & 2 & 2,3 & 2,3 & $(1,2,2,0,0,0)$ & $(1,3,4,\text{--},\text{--},\text{--})$ \\
3 & 2 & 2,3 & 3 & 2 & $(1,3,3,0,0,0)$ & $(1,1,2,\text{--},\text{--},\text{--})$ \\
4 & -- & 2 & 2,3 & 2,3 & $(1,4,4,0,0,0)$ & $(1,3,4,\text{--},\text{--},\text{--})$ \\
5 & 2 & 2,3 & 3 & 2 & $(1,5,5,0,0,0)$ & $(1,1,2,\text{--},\text{--},\text{--})$ \\
6 & -- & 2 & 2,3 & 1,2 & $(2,6,6,0,0,0)$ & $(4,3,3,\text{--},\text{--},\text{--})$ \\
7 & 1,2 & 2,3 & -- & 2 & $(3,7,7,0,0,0)$ & $(1,4,2,\text{--},\text{--},\text{--})$ \\
8 & 2 & 2 & 2,3 & 1 & $(4,8,8,0,0,0)$ & $(4,1,3,\text{--},\text{--},\text{--})$ \\
9 & 1 & 2,3 & 2 & 2 & $(5,9,9,0,0,0)$ & $(1,2,2,\text{--},\text{--},\text{--})$ \\
10 & 2 & -- & 2,3 & 1,2 & $(6,10,10,0,0,0)$ & $(4,3,3,\text{--},\text{--},\text{--})$ \\
11 & 1,2 & 2,3 & -- & 2 & $(7,11,11,0,0,0)$ & $(1,4,2,\text{--},\text{--},\text{--})$ \\
12 & 2 & 2 & 2,3 & 1 & $(8,12,12,0,0,0)$ & $(4,1,3,\text{--},\text{--},\text{--})$ \\
13 & 1 & 2,3 & 2 & 2 & $(9,13,13,0,0,0)$ & $(1,2,2,\text{--},\text{--},\text{--})$ \\
14 & 2 & -- & 2,3 & 1,2 & $(10,14,14,0,0,0)$ & $(4,3,3,\text{--},\text{--},\text{--})$ \\
15 & 2,4 & 2,3 & -- & 2 & $(10,15,15,1,0,0)$ & $(4,4,2,1,\text{--},\text{--})$ \\
16 & 2 & 2 & 2,3 & -- & $(10,16,16,1,0,0)$ & $(4,1,3,1,\text{--},\text{--})$ \\
17 & -- & 2,3 & 4 & 2 & $(10,17,17,2,0,0)$ & $(4,2,2,3,\text{--},\text{--})$ \\
18 & 1,2 & -- & 3 & 2 & $(11,18,18,2,0,0)$ & $(1,4,3,3,\text{--},\text{--})$ \\
19 & 2 & 2,3 & -- & 1 & $(12,19,19,2,0,0)$ & $(4,1,2,3,\text{--},\text{--})$ \\
20 & 4 & 2 & 3 & 2 & $(12,20,20,3,0,0)$ & $(4,2,3,1,\text{--},\text{--})$ \\
21 & 2 & 3 & 6 & 2 & $(12,21,21,3,0,1)$ & $(4,4,2,1,\text{--},3)$ \\
22 & 2 & 2 & 3 & -- & $(12,22,22,3,0,1)$ & $(4,1,3,1,\text{--},3)$ \\
23 & -- & 2,3 & -- & 2 & $(12,23,23,3,0,1)$ & $(4,2,2,1,\text{--},3)$ \\
24 & 2 & -- & 3 & 2 & $(12,24,24,3,0,1)$ & $(4,4,3,1,\text{--},3)$ \\
25 & 2 & 2,3 & -- & -- & $(12,25,25,3,0,1)$ & $(4,1,2,1,\text{--},3)$ \\
26 & -- & 2 & 3 & 2 & $(12,26,26,3,0,1)$ & $(4,2,3,1,\text{--},3)$ \\
27 & 2 & 3 & -- & 2 & $(12,27,27,3,0,1)$ & $(4,4,2,1,\text{--},3)$ \\
28 & 2 & 2 & 3 & -- & $(12,28,28,3,0,1)$ & $(4,1,3,1,\text{--},3)$ \\
29 & -- & 2,3 & -- & 2 & $(12,29,29,3,0,1)$ & $(4,2,2,1,\text{--},3)$ \\
30 & 1 & -- & 3 & 2 & $(13,30,30,3,0,1)$ & $(1,4,3,1,\text{--},3)$ \\
31 & -- & 2,3 & -- & 1 & $(14,31,31,3,0,1)$ & $(4,2,2,1,\text{--},3)$ \\
32 & 1 & -- & 3 & 2 & $(15,32,32,3,0,1)$ & $(1,4,3,1,\text{--},3)$ \\
33 & -- & 2,3 & -- & 4 & $(15,33,33,4,0,1)$ & $(1,2,2,4,\text{--},3)$ \\
34 & 4 & -- & 3 & 2 & $(15,34,34,5,0,1)$ & $(1,4,3,1,\text{--},3)$ \\
35 & -- & 2,3 & -- & 1 & $(16,35,35,5,0,1)$ & $(4,2,2,1,\text{--},3)$ \\
36 & 1 & -- & 3 & 2 & $(17,36,36,5,0,1)$ & $(1,4,3,1,\text{--},3)$ \\
37 & -- & 2,6 & -- & 4 & $(17,37,36,6,0,2)$ & $(1,2,3,4,\text{--},2)$ \\
38 & 4 & -- & 6 & 2 & $(17,38,36,7,0,3)$ & $(1,4,3,1,\text{--},3)$ \\
39 & -- & 2,3 & -- & 1 & $(18,39,37,7,0,3)$ & $(4,2,2,1,\text{--},3)$ \\
40 & 1 & -- & 3 & 2 & $(19,40,38,7,0,3)$ & $(1,4,3,1,\text{--},3)$ \\
41 & -- & 2,6 & -- & 4 & $(19,41,38,8,0,4)$ & $(1,2,3,4,\text{--},2)$ \\
42 & 4 & -- & 6 & 2 & $(19,42,38,9,0,5)$ & $(1,4,3,1,\text{--},3)$ \\
43 & -- & 2,3 & -- & 1 & $(20,43,39,9,0,5)$ & $(4,2,2,1,\text{--},3)$ \\
44 & 1 & -- & 3 & 2 & $(21,44,40,9,0,5)$ & $(1,4,3,1,\text{--},3)$ \\
45 & -- & 2,5 & -- & 4 & $(21,45,40,10,1,5)$ & $(1,2,3,4,2,3)$ \\
46 & 4 & -- & -- & 2 & $(21,46,40,11,1,5)$ & $(1,4,3,1,2,3)$ \\
47 & -- & 2 & -- & 1 & $(22,47,40,11,1,5)$ & $(4,2,3,1,2,3)$ \\
48 & 1 & -- & -- & 2 & $(23,48,40,11,1,5)$ & $(1,4,3,1,2,3)$ \\
49 & -- & 2 & -- & 4 & $(23,49,40,12,1,5)$ & $(1,2,3,4,2,3)$ \\
50 & 4 & -- & -- & 2 & $(23,50,40,13,1,5)$ & $(1,4,3,1,2,3)$ \\
51 & -- & 2 & -- & 1 & $(24,51,40,13,1,5)$ & $(4,2,3,1,2,3)$ \\
52 & 1 & -- & -- & 2 & $(25,52,40,13,1,5)$ & $(1,4,3,1,2,3)$ \\
53 & -- & 2 & -- & 4 & $(25,53,40,14,1,5)$ & $(1,2,3,4,2,3)$ \\
54 & 4 & -- & -- & 1 & $(26,53,40,15,1,5)$ & $(4,2,3,1,2,3)$ \\
55 & 1 & -- & -- & 2 & $(27,54,40,15,1,5)$ & $(1,4,3,1,2,3)$ \\
56 & -- & 2 & -- & 4 & $(27,55,40,16,1,5)$ & $(1,2,3,4,2,3)$ \\
57 & 4 & -- & -- & 1 & $(28,55,40,17,1,5)$ & $(4,2,3,1,2,3)$ \\
58 & 1 & -- & -- & 2 & $(29,56,40,17,1,5)$ & $(1,4,3,1,2,3)$ \\
59 & -- & 2 & -- & 4 & $(29,57,40,18,1,5)$ & $(1,2,3,4,2,3)$ \\
60 & 4 & -- & -- & 1 & $(30,57,40,19,1,5)$ & $(4,2,3,1,2,3)$ \\
61 & 1 & -- & -- & 2 & $(31,58,40,19,1,5)$ & $(1,4,3,1,2,3)$ \\
62 & -- & 2 & -- & 4 & $(31,59,40,20,1,5)$ & $(1,2,3,4,2,3)$ \\
63 & 4 & -- & -- & 1 & $(32,59,40,21,1,5)$ & $(4,2,3,1,2,3)$ \\
64 & 1 & -- & -- & 2 & $(33,60,40,21,1,5)$ & $(1,4,3,1,2,3)$ \\
65 & -- & 2 & -- & 4 & $(33,61,40,22,1,5)$ & $(1,2,3,4,2,3)$ \\
66 & 4 & -- & -- & 1 & $(34,61,40,23,1,5)$ & $(4,2,3,1,2,3)$ \\
67 & 1 & -- & -- & 2 & $(35,62,40,23,1,5)$ & $(1,4,3,1,2,3)$ \\
68 & -- & 2 & -- & 4 & $(35,63,40,24,1,5)$ & $(1,2,3,4,2,3)$ \\
69 & 4 & -- & -- & 1 & $(36,63,40,25,1,5)$ & $(4,2,3,1,2,3)$ \\
70 & 1 & -- & -- & 6 & $(37,63,40,25,1,6)$ & $(1,2,3,1,2,4)$ \\
71 & -- & -- & 6 & 2 & $(37,64,40,25,1,7)$ & $(1,4,3,1,2,3)$ \\
72 & -- & 2 & -- & 4 & $(37,65,40,26,1,7)$ & $(1,2,3,4,2,3)$ \\
73 & 4 & -- & -- & 1 & $(38,65,40,27,1,7)$ & $(4,2,3,1,2,3)$ \\
74 & 1 & -- & -- & 3 & $(39,65,41,27,1,7)$ & $(1,2,4,1,2,3)$ \\
75 & -- & -- & 3 & 6 & $(39,65,42,27,1,8)$ & $(1,2,3,1,2,4)$ \\
76 & -- & -- & 6 & 5 & $(39,65,42,27,2,9)$ & $(1,2,3,1,4,3)$ \\
77 & -- & 5 & -- & 2 & $(39,66,42,27,3,9)$ & $(1,4,3,1,2,3)$ \\
78 & -- & 2 & -- & 4 & $(39,67,42,28,3,9)$ & $(1,2,3,4,2,3)$ \\
79 & 4 & -- & -- & 1 & $(40,67,42,29,3,9)$ & $(4,2,3,1,2,3)$ \\
80 & 1 & -- & -- & 3 & $(41,67,43,29,3,9)$ & $(1,2,4,1,2,3)$ \\
81 & -- & -- & 3 & 6 & $(41,67,44,29,3,10)$ & $(1,2,3,1,2,4)$ \\
82 & -- & -- & 6 & 5 & $(41,67,44,29,4,11)$ & $(1,2,3,1,4,3)$ \\
83 & -- & 5 & -- & 2 & $(41,68,44,29,5,11)$ & $(1,4,3,1,2,3)$ \\
84 & -- & 2 & -- & 4 & $(41,69,44,30,5,11)$ & $(1,2,3,4,2,3)$ \\
85 & 4 & -- & -- & 1 & $(42,69,44,31,5,11)$ & $(4,2,3,1,2,3)$ \\
86 & 1 & -- & -- & 3 & $(43,69,45,31,5,11)$ & $(1,2,4,1,2,3)$ \\
87 & -- & -- & 3 & 6 & $(43,69,46,31,5,12)$ & $(1,2,3,1,2,4)$ \\
88 & -- & -- & 6 & 5 & $(43,69,46,31,6,13)$ & $(1,2,3,1,4,3)$ \\
89 & -- & 5 & -- & 2 & $(43,70,46,31,7,13)$ & $(1,4,3,1,2,3)$ \\
90 & -- & 2 & -- & 4 & $(43,71,46,32,7,13)$ & $(1,2,3,4,2,3)$ \\
91 & 4 & -- & -- & 1 & $(44,71,46,33,7,13)$ & $(4,2,3,1,2,3)$ \\
92 & 1 & -- & -- & 3 & $(45,71,47,33,7,13)$ & $(1,2,4,1,2,3)$ \\
93 & -- & -- & 3 & 6 & $(45,71,48,33,7,14)$ & $(1,2,3,1,2,4)$ \\
94 & -- & -- & 6 & 5 & $(45,71,48,33,8,15)$ & $(1,2,3,1,4,3)$ \\
95 & -- & 5 & -- & 2 & $(45,72,48,33,9,15)$ & $(1,4,3,1,2,3)$ \\
96 & -- & 2 & -- & 4 & $(45,73,48,34,9,15)$ & $(1,2,3,4,2,3)$ \\
97 & 4 & -- & -- & 1 & $(46,73,48,35,9,15)$ & $(4,2,3,1,2,3)$ \\
98 & 1 & -- & -- & 3 & $(47,73,49,35,9,15)$ & $(1,2,4,1,2,3)$ \\
99 & -- & -- & 3 & 6 & $(47,73,50,35,9,16)$ & $(1,2,3,1,2,4)$ \\
100 & -- & -- & 6 & 5 & $(47,73,50,35,10,17)$ & $(1,2,3,1,4,3)$ \\
101 & -- & 5 & -- & 2 & $(47,74,50,35,11,17)$ & $(1,4,3,1,2,3)$ \\
102 & -- & 2 & -- & 4 & $(47,75,50,36,11,17)$ & $(1,2,3,4,2,3)$ \\
103 & 4 & -- & -- & 1 & $(48,75,50,37,11,17)$ & $(4,2,3,1,2,3)$ \\
104 & 1 & -- & -- & 3 & $(49,75,51,37,11,17)$ & $(1,2,4,1,2,3)$ \\
105 & -- & -- & 3 & 6 & $(49,75,52,37,11,18)$ & $(1,2,3,1,2,4)$ \\
106 & -- & -- & 6 & 5 & $(49,75,52,37,12,19)$ & $(1,2,3,1,4,3)$ \\
107 & -- & 5 & -- & 2 & $(49,76,52,37,13,19)$ & $(1,4,3,1,2,3)$ \\
108 & -- & 2 & -- & 4 & $(49,77,52,38,13,19)$ & $(1,2,3,4,2,3)$ \\
109 & 4 & -- & -- & 1 & $(50,77,52,39,13,19)$ & $(4,2,3,1,2,3)$ \\
110 & 1 & -- & -- & 3 & $(51,77,53,39,13,19)$ & $(1,2,4,1,2,3)$ \\
111 & -- & -- & 3 & 6 & $(51,77,54,39,13,20)$ & $(1,2,3,1,2,4)$ \\
112 & -- & -- & 6 & 5 & $(51,77,54,39,14,21)$ & $(1,2,3,1,4,3)$ \\
113 & -- & 5 & -- & 2 & $(51,78,54,39,15,21)$ & $(1,4,3,1,2,3)$ \\
114 & -- & 2 & -- & 4 & $(51,79,54,40,15,21)$ & $(1,2,3,4,2,3)$ \\
115 & 4 & -- & -- & 1 & $(52,79,54,41,15,21)$ & $(4,2,3,1,2,3)$ \\
116 & 1 & -- & -- & 3 & $(53,79,55,41,15,21)$ & $(1,2,4,1,2,3)$ \\
117 & -- & -- & 3 & 6 & $(53,79,56,41,15,22)$ & $(1,2,3,1,2,4)$ \\
118 & -- & -- & 6 & 5 & $(53,79,56,41,16,23)$ & $(1,2,3,1,4,3)$ \\
119 & -- & 5 & -- & -- & $(53,79,56,41,17,23)$ & $(1,2,3,1,2,3)$ \\
120 & -- & -- & -- & -- & $(53,79,56,41,17,23)$ & $(1,2,3,1,2,3)$ \\
\end{longtable}
\endgroup

\newpage
\emph{Competitive equilibrium.} Table~\ref{tab:app-ce} verifies
Definition~\ref{defCE} under the modified valuations of
Proposition~\ref{epsilonoptimal}. For each bidder, it reports the modified
surplus of the final allocated bundle and the maximum modified surplus over all
$2^{|K|}$ bundles; the two coincide in every case, so~\eqref{ce1} holds. Any
product left unallocated carries a zero standing price, so~\eqref{ce2} holds as
well.

\begin{table}[H]
\centering\small
\caption{Competitive-equilibrium verification under the modified valuations.}
\label{tab:app-ce}
\setlength{\tabcolsep}{8pt}
\renewcommand{\arraystretch}{1.12}
\begin{tabular}{@{}lccrrc@{}}
\toprule
Valuation & $b$ & $U_b^*$ & \shortstack{Allocated\\modified surplus} & \shortstack{Maximum\\modified surplus} & Maximizer \\
\midrule
Additive & 1 & $\{1,4\}$ & 15.68 & 15.68 & Yes \\
 & 2 & $\{2,5\}$ & 24.93 & 24.93 & Yes \\
 & 3 & $\{3,6\}$ & 27.52 & 27.52 & Yes \\
 & 4 & $\emptyset$ & 0.00 & 0.00 & Yes \\
\addlinespace[4pt]
Symmetric concave & 1 & $\{5,6\}$ & 22.52 & 22.52 & Yes \\
 & 2 & $\{3,4\}$ & 48.45 & 48.45 & Yes \\
 & 3 & $\{1\}$ & 24.60 & 24.60 & Yes \\
 & 4 & $\{2\}$ & 24.00 & 24.00 & Yes \\
\addlinespace[4pt]
Unit-demand & 1 & $\{4\}$ & 49.16 & 49.16 & Yes \\
 & 2 & $\{2\}$ & 63.00 & 63.00 & Yes \\
 & 3 & $\{3\}$ & 67.60 & 67.60 & Yes \\
 & 4 & $\{1\}$ & 43.57 & 43.57 & Yes \\
\addlinespace[4pt]
$K_b$-demand & 1 & $\{1,4\}$ & 15.68 & 15.68 & Yes \\
 & 2 & $\{2,5\}$ & 25.93 & 25.93 & Yes \\
 & 3 & $\{3,6\}$ & 35.52 & 35.52 & Yes \\
 & 4 & $\emptyset$ & 0.00 & 0.00 & Yes \\
\bottomrule
\end{tabular}
\end{table}

\end{document}